%% file: Random_Many_access.tex
\documentclass[11pt, onecolumn]{IEEEtran}

\usepackage[utf8]{inputenc}
\usepackage[T1]{fontenc}
\usepackage{url}
\usepackage{ifthen}
\usepackage{cite}
\usepackage[cmex10]{amsmath} 
\usepackage{amsmath}
\usepackage{amsfonts}
\usepackage{amssymb}

\usepackage{psfrag}
\usepackage{enumerate}
\usepackage{url}

\usepackage{graphicx}
\usepackage{color}

\newtheorem{theorem}{Theorem}
\newtheorem{lemma}{Lemma}

\newtheorem{definition}{Definition}

\newcommand{\comment}[1]{}

\def\ISIT{0}

\def\bX{{\bf X}}
\def\bY{{\bf Y}}
\def\bZ{{\bf Z}}
\def\bS{{\bf S}}
\def\bs{{\bf s}}
\def\bD{{\bf D}}
\def\bd{{\bf d}}

\def\bq{{\bf q}}

\def\bx{{\bf x}}
\def\by{{\bf y}}

\def\bW{{\bf W}}
\def\bhW{{\bf \hat{W}}}

\def\CR{\dot{R}}
\def\CC{\dot{C}}
\def\CCP{\dot{C}_{\bot}}

\def\by{{\bf y}}

\def\bc{{\bf c}}
\def\bw{{\bf w}}

\def\CR{\dot{R}}

\def\CC{\dot{C}}

\def\PR{\dot{R^{A}}}

\def\CCP{\dot{C}_{\bot}}

\def\el{\ell}

\def\cA{\mbox{$\cal{A}$}}

\def\cC{\mbox{$\cal{C}$}}
\def\cD{\mbox{$\cal{D }$}}

\def\cE{\mbox{$\cal{E }$}}

\def\cB{\mbox{$\cal{B}$}}

\def\cY{\mbox{$\cal{Y}$}}

\def\cW{\mbox{$\cal{W}$}}

\def\cK{\mbox{$\cal{K}$}}
\def\cT{\mbox{$\cal{T}$}}

\def\Pr{\text{Pr} }

\def\cS{\mbox{$\cal{S}$}}

\def\cW{\mbox{$\cal{W}$}}
\def\cE{\mbox{$\cal{E}$}}

\def\cK{\mbox{$\cal{K}$}}

\def\cN{\mbox{$\cal{N}$}}

\newcommand{\I}[1]{\mathbf{1} (#1)}

\begin{document}
	\IEEEoverridecommandlockouts
	\title{Capacity per Unit-Energy of \\Gaussian Random  Many-Access  Channels
	}
	\author{
		\IEEEauthorblockN{Jithin~Ravi\IEEEauthorrefmark{2}\IEEEauthorrefmark{3} and Tobias~Koch\IEEEauthorrefmark{2}\IEEEauthorrefmark{3}}\\
		\IEEEauthorblockA{\IEEEauthorrefmark{2}%
			Signal Theory and Communications Department, Universidad Carlos III de Madrid, 28911, Legan\'es, Spain\\
			\IEEEauthorrefmark{3}%
			Gregorio Mara\~n\'on Health Research Institute, 28007, Madrid, Spain.\\
			Emails: {\{rjithin,koch\}@tsc.uc3m.es}}
		\thanks{J.~Ravi and T.~Koch have received funding from the European Research Council (ERC) under the European Union's Horizon 2020 research and innovation programme (Grant No.~714161). T.~Koch has further received funding from the Spanish Ministerio de Econom\'ia y Competitividad under Grants RYC-2014-16332 and TEC2016-78434-C3-3-R (AEI/FEDER, EU).}
	}
	
	\maketitle
	\begin{abstract}
	We consider a Gaussian multiple-access channel  with random user activity where the total number of users $\ell_n$ and the average number of active users $k_n$ may be unbounded. For this channel, we characterize the maximum number of bits that can be transmitted reliably per unit-energy in terms of $\ell_n$ and $k_n$. We show that if $k_n\log \ell_n$ is sublinear in $n$, then each user can achieve the single-user capacity per unit-energy. Conversely, if $k_n\log \ell_n$ is superlinear in $n$, then the capacity per unit-energy is zero. We further demonstrate that orthogonal-access schemes, which are optimal when all users are active with probability one, can be strictly suboptimal.
\end{abstract}

\section{Introduction}
Chen \emph{et al.} \cite{ChenCG17} introduced the many-access channel (MnAC) as a multiple-access channel (MAC) where the number of users grows with the blocklength and each user is active with a given probability. This model is motivated by systems consisting of a single receiver and many transmitters, the number of which is comparable or even larger than the blocklength, a situation that may occur, \emph{e.g.}, in a machine-to-machine communication system with many thousands of devices in a given cell that are active only sporadically. In \cite{ChenCG17}, Chen \emph{et al.} considered a Gaussian MnAC with $\el_n$ users, each of which is active with probability $\alpha_n$, and determined the number of messages $M_n$ each user can transmit reliably with a codebook of average power not exceeding $P$. Since then, MnACs have been studied in various papers under different settings. For example, Polyanskiy \cite{Polyanskiy17} considered a Gaussian MnAC where the number of active users grows linearly in the blocklength and each user's payload is fixed. Zadik \emph{et al.} \cite{ZadikPT19} presented improved bounds on the tradeoff between user density and energy-per-bit of this channel. Low-complexity schemes for the MnAC were studied in \cite{OrdentlichP17,VemNCC17}. Generalizations to quasi-static fading MnACs can be found in \cite{KowshikPISIT19,KowshikP19,KowshiKAFPISIT19,KowshiKAFP19}. Shahi \emph{et al.} \cite{ShahiTD18} studied the capacity region of strongly asynchronous MnACs.

Recently, we studied the capacity per unit-energy of the Gaussian MnAC as a function of the order of growth of users when all users are active with probability one \cite{RaviKISIT19}. We showed that if the order of growth is above $n/ \log n$, then the capacity per unit-energy is zero, and if the order of growth is below $n/ \log n$, then each user can achieve the singe-user capacity per unit-energy. Thus, there is a sharp transition between orders of growth where interference-free communication is feasible and orders of growth where reliable communication at a positive rate is infeasible. We further showed that the capacity per unit-energy can be achieved by an \emph{orthogonal-access scheme} where the codewords of different users are orthogonal to each other.

In this paper, we extend the analysis of \cite{RaviKISIT19} to a random-access setting. In particular, we consider a setting where the total number of users $\el_n$ may grow as an arbitrary function of the blocklength and the probability $\alpha_n$ that a user is active may be a function of the blocklength, too.
Let $k_n = \alpha_n \el_n$ denote the average number of active users. 
We demonstrate that if $k_n \log \el_n$ is sublinear in $n$, then each user can achieve the single-user capacity per unit-energy. Conversely, if $k_n \log \el_n$ is superlinear in $n$, then the capacity per unit-energy is zero. Hence, there is again a sharp transition between orders of growth where interference-free communication is feasible and orders of growth where reliable communication at a positive rate is infeasible, but the transition threshold depends on the behaviors of both $\el_n$ and $k_n$.
We further show that orthogonal-access schemes, which are optimal when $\alpha_n=1$, are strictly suboptimal when $\alpha_n \to 0$.

The rest of the paper is organized as follows. Section~\ref{Sec_model} introduces the system model. Section~\ref{sec_joint} presents our main results. Section~\ref{sec_average} briefly discusses the capacity per unit-energy when the error probability is replaced by the so-called \emph{per-user probability of error} considered, e.g., in \cite{Polyanskiy17, OrdentlichP17,VemNCC17,ZadikPT19,KowshikPISIT19,KowshikP19,KowshiKAFPISIT19,KowshiKAFP19}.

\section{Problem Formulation and Preliminaries}
\label{Sec_model}
\subsection{Model and Definitions}
\label{Sec_Def}
Consider a network with $\el$ users that, if they are active, wish to transmit their messages $W_i, i=1, \ldots, \el$ to one common receiver. The messages are assumed to be independent and uniformly distributed on $\{1,\ldots,M_n^{(i)}\}$. To transmit their messages, the users send a codeword of $n$ symbols over the channel, where $n$ is referred to as the \emph{blocklength}. We consider a many-access scenario where the number of users $\el$ grows with $n$, hence, we denote it as $\el_n$.  We further assume that a user is active with probability $\alpha_n$, where $\alpha_n \to \alpha \in [0,1]$ as $n$ tends to infinity.
Since an inactive user is equivalent to a user transmitting the all-zero codeword, we can express the distribution of the $i$-th user's message as
\begin{align}
\Pr\{W_i = w\}  =
\begin{cases}
1 - \alpha_n,  & \quad w=0 \\
\frac{\alpha_n}{M_n^{(i)}}, & \quad w \in \{1,\ldots,M_n^{(i)}\}
\end{cases}
\label{Eq_messge_def}
\end{align}
and assume that the codebook is such that message $0$ is mapped to the all-zero codeword. We denote the average number of active users at blocklength $n$  by $k_n$, i.e., $k_n = \alpha_n \el_n$.

We  consider a Gaussian channel model where the received vector $\bY$ is given by
\begin{align*}
\bY & = \sum_{i=1}^{\el_n} \bX_i(W_i) + \bZ. 
\end{align*}
Here $ \bX_i(W_i)$ is the $n$-length transmitted codeword from user $i$ for message $W_i$ and $\bZ$ is 
a vector of $n$ i.i.d. Gaussian components $Z_j \sim \cN(0, N_0/2)$ independent of $\bX_i$.

\begin{definition}
	\label{Def_nMCode}
	For $0 \leq \epsilon < 1$, an  $(n,\bigl\{M_n^{(\cdot)}\bigr\},\bigl\{E_n^{(\cdot)}\bigr\}, \epsilon)$ code for the Gaussian many-access channel consists of:
	\begin{enumerate}
		\item Encoding functions $f_i: \{0, 1,\ldots,M_n^{(i)}\} \rightarrow \mathbb{R}^n$, \mbox{$i =1,\ldots, \el_n$} which map user $i$'s message to the codeword $\bX_i(W_i)$, satisfying the energy constraint
		\begin{align}
		\label{Eq_energy_consrnt}
		\sum_{j=1}^{n} x_{ij}^2(w_i) \leq E_n^{(i)}
		\end{align}
		where $x_{ij}$ is the $j$-th symbol of the transmitted codeword. If $W_i =0$, then $x_{ij} =0$ for $j=1, \ldots, n$.
		\item Decoding function $g: \mathbb{R}^n \rightarrow \{ 0,1,\ldots,M_n^{(1)}\} \times \ldots \times \{ 0,1,\ldots,M_n^{(\el_n)}\} $ which maps the received vector $\bY$ to the messages of all users and whose probability of error $P_{e}^{(n)}$ satisfies
	\end{enumerate} 
	\begin{align}
	\label{Eq_prob_err}
	P_{e}^{(n)} \triangleq  \Pr\{ g(\bY) \neq (W_1,\ldots,W_{\el_n}) \} \leq \epsilon.
	\end{align}
\end{definition}

An $(n,\{M_n^{(\cdot)}\},\{E_n^{(\cdot)}\}, \epsilon)$ code is said to be \emph{symmetric} if $M_n^{(i)} = M_n$ and $E_n^{(i)} = E_n$ for all $i=1, \ldots, \el_n$. For compactness, we denote such a code by $(n, M_n, E_n, \epsilon)$. In this paper, we restrict ourselves to symmetric codes.

\begin{definition}
	\label{Def_Sym_Rate_Cost}
	For a symmetric code, the rate per unit-energy $\CR$ is said to be $\epsilon$-achievable if for every $\delta > 0$ there exists an $n_0$ such that if $n \geq n_0$, then an $(n,M_n,E_n, \epsilon)$ code can be found whose rate per unit-energy satisfies $\frac{\log M_n}{ E_n} > \CR - \delta$. Furthermore, $\CR$ is said to be achievable if it is $\epsilon$-achievable for all $0 < \epsilon < 1$.  The capacity per unit-energy $\CC$ is the supremum of all achievable rates per unit-energy.
\end{definition}

\if \ISIT 1

\subsection{Order Notations}
Let $\{a_n\}$ and $\{b_n\}$ be two sequences of nonnegative real numbers.
We write $a_n = O(b_n)$  if there exists an $n_0$ and a positive real number $S$ such that for all $n \geq n_0$, $a_n \leq S b_n$. We write $a_n = o(b_n)$ if $ \lim\limits_{n\rightarrow \infty} \frac{a_n}{b_n} = 0$, and $a_n = \Omega(b_n)$ if $\liminf\limits_{n \rightarrow \infty} \frac{a_n}{b_n} >0$. 
Similarly, $a_n = \Theta (b_n)$ indicates that there exist $ 0 < l_1<l_2$ and $n_0$ such that 
$l_1 b_n \leq a_n \leq l_2 b_n$
for all $n \geq n_0$.  
We  write
$a_n = \omega (b_n)$ if $\lim\limits_{n\rightarrow \infty}  \frac{a_n}{b_n} = \infty$.

\else

\subsection{Order Notations}
Let $\{a_n\}$ and $\{b_n\}$ be two sequences of nonnegative real numbers.
We write $a_n = O(b_n)$  if there exists an $n_0$ and a positive real number $S$ such that for all $n \geq n_0$, $a_n \leq S b_n$. We write $a_n = o(b_n)$ if $ \lim\limits_{n\rightarrow \infty} \frac{a_n}{b_n} = 0$, and $a_n = \Omega(b_n)$ if $\liminf\limits_{n \rightarrow \infty} \frac{a_n}{b_n} >0$. 
Similarly, $a_n = \Theta (b_n)$ indicates that there exist $ 0 < l_1<l_2$ and $n_0$ such that 
$l_1 b_n \leq a_n \leq l_2 b_n$
for all $n \geq n_0$.  
We further write
$a_n = \omega (b_n)$ if $\lim\limits_{n\rightarrow \infty}  \frac{a_n}{b_n} = \infty$.

\fi

\section{Capacity per Unit-Energy}
\label{sec_joint}
In this section, we discuss our results on the behavior of capacity per unit-energy for Gaussian random MnACs. Our main result is 
Theorem~\ref{Thm_random_JPE}, which characterizes the capacity per unit-energy in terms of $\el_n$ and $k_n$.  In Theorem~\ref{Thm_ortho_accs}, we characterize the behavior of the largest rate  per unit-energy that can be achieved by an orthogonal-access scheme. These results are presented in Subsection~\ref{Sec_results}. The proofs of Theorems~\ref{Thm_random_JPE} and~\ref{Thm_ortho_accs} are given in Subsections~\ref{Sec_proof_JPE} and~\ref{Sec_proof_ortho_access}, respectively.

Before presenting our results, we first note that the case where $k_n$ vanishes as $n \to \infty$ is uninteresting. 
Indeed, this case only happens if $\alpha_n \to 0$. 
Then, the probability that all the users are inactive, given by $\bigl( (1-\alpha_n)^{\frac{1}{\alpha_n}}\bigr)^{k_n}$,
tends to one since  $(1-\alpha_n)^{\frac{1}{\alpha_n}} \to 1/e $ and $k_n \to 0$.  Consequently, a code with $M_n=2$ and $E_n =0$ for all $n$ and a decoding function that always declares that all users are inactive achieve an error probability
$P_{e}^{(n)}$ that vanishes as $n \to \infty$. This implies that $\CC= \infty$. In the following, we avoid this trivial case and assume that $\el_n$ and $\alpha_n$ are such that $k_n$ is bounded away from zero.

\subsection{Our Main Results}
\label{Sec_results}
\begin{theorem}
	\label{Thm_random_JPE}
	Assume that $k_n =\Omega(1)$. Then the capacity per unit-energy of the Gaussian random MnAC has the following behavior:
	\begin{enumerate}
		\item 	If $k_n \log \el_n = o(n)$,   then $\CC = (\log e )/ N_0$. \label{Thm_achv_part}
		\item 	If $k_n \log \el_n = \omega(n)$, then $\CC =0$. 	\label{Thm_conv_part}
	\end{enumerate}
\end{theorem}
\begin{proof}
	See Subsection~\ref{Sec_proof_JPE}.
\end{proof}

Theorem~\ref{Thm_random_JPE} demonstrates that there is a sharp transition between orders of growth where interference-free communication is feasible and orders of growth where no positive rate per unit-energy is feasible. The same behavior was observed for the non-random-access case, where the transition threshold seperating these two regimes is at $n/ \log n$~\cite{RaviKISIT19}. When $\alpha_n$ converges to a positive value, the order of growth of $k_n \log \el_n$ coincides with that of both $k_n \log k_n$ and  $\el_n \log \el_n$. In this case, the transition threshold in the random-access case is also at $n /\log n$.
However, when $\alpha_n \to 0$, the orders of growth of $k_n$ and $\el_n$ are different and the transition threshold for $\el_n$ is in general larger than $n / \log n$, so random user-activity enables interference-free communication at an order of growth above the limit $n/ \log n$ of the non-random-access case. Similarly, when $\alpha_n \to 0$, the transition threshold for $k_n$ is in general smaller than $n/ \log n$, so treating a random MnAC with $\el_n$ users as a non-random MnAC with $k_n$ users may be overly-optimistic.

In~\cite{RaviKISIT19}, it was shown that, when \mbox{$k_n =o(n/ \log n)$} and \mbox{$\alpha_n =1$}, an orthogonal-access scheme is sufficient to achieve the capacity per unit-energy. It turns out that this is not the case anymore when $\alpha_n \to 0$.
\begin{theorem}
	\label{Thm_ortho_accs}
	Assume that $k_n = \Omega(1)$. The largest rate per unit-energy $\CCP$ achievable with an orthogonal-access scheme satisfies the following:
	\begin{enumerate}[1)]
		\item  If  $ \el_n = o(n/ \log n)$, then $\CCP = (\log e )/ N_0$. \label{Thm_ortho_accs_achv}
		\item If $ \el_n = \omega(n/ \log n)$, then $\CCP =0$. \label{Thm_ortho_accs_conv}
	\end{enumerate}
\end{theorem}
\begin{proof}
	See Subsection~\ref{Sec_proof_ortho_access}.
\end{proof}

Observe that there is again a sharp transition between the orders of growth of $\el_n$ where interference-free communication is feasible and orders of growth where no positive rate per unit-energy is feasible. In contrast to the optimal transmission scheme, the transition threshold for orthogonal-access schemes happens at $n/ \log n$, irrespective of the behavior of $\alpha_n$. Thus, by using an orthogonal-access scheme, we treat the random MnAC as if it were a non-random MnAC. Theorem~\ref{Thm_ortho_accs} also implies that there are orders of growth of $\el_n$ and $k_n$ where non-orthogonal-access schemes are necessary to achieve the capacity per unit-energy.

\subsection{Proof of Theorem~\ref{Thm_random_JPE}}
\label{Sec_proof_JPE}
To prove Part~\ref{Thm_achv_part}), we use an achievability scheme  with a decoding process consisting of two steps.  First, the receiver determines which users are active. If the number of estimated active users is less than or equal to $\xi k_n$ for some positive integer $\xi$, then the receiver decodes the messages of all active users. If the number of estimated active users is greater than $\xi k_n$, then it declares an error. The total error probability of this scheme is upper-bounded by
\begin{equation*}
P(\cD) + \sum_{k'_n=1}^{\xi k_n}\Pr\{K'_n=k_n'\}P\bigl(\cE_m(k'_n)\bigr) + \Pr\{K'_n>\xi k_n\}
\end{equation*}
where $K'_n$ is the number of active users, $P(\cD)$ is the probability of a detection error, and $P\bigl(\cE_m(k'_n)\bigr)$ is the probability of a decoding error when the receiver has correctly detected that there are $k'_n$ users active. In the following, we show that these probabilities vanish as $n\to\infty$ for any fixed, positive integer $\xi$. Furthermore, by Markov's inequality, we have that $\Pr\{K'_n>\xi k_n\}\leq 1/\xi$. It thus follows that the total probability of error vanishes as we let first $n\to\infty$ and then $\xi\to\infty$.

To enable user detection at the receiver, out of $n$ channel uses, each user uses the first $n''$ channel uses to send its signature and  \mbox{$n'=n -n''$} channel uses for sending the message. Furthermore, the
signature uses energy $E_n''$ out of $E_n$, while the energy used for sending message is given by $E_n' = E_n -E_n''$.  

Let  $\bs_i$ denote the signature of  user $i$ and $\tilde{\bx}_i(w_i)$ denote the codeword of length $n'$ for sending the message $w_i$, where $w_i =1,\ldots, M_n$. Then the codeword $\bx_i(w_i)$ is given by
\begin{align*}
\bx_i(w_i) = (\bs_i,  \tilde{\bx}_i(w_i)).
\end{align*}
Explicitly, for a given arbitrary $0 < b < 1$, we let
\begin{equation}
n'' = bn, \quad \label{Eq_channel_choice}
\end{equation}
and 
\begin{equation}
\label{Eq_energy_choice}
E_n'' = bE_n,  \quad E_n = c_n \ln \el_n 
\end{equation}
with $c_n = \ln (\frac{n}{k_n\ln \el_n})$.

Based on the first $n''$ received symbols, the receiver detects which users are active. We need the following lemma to show that the detection error probability  vanishes as $n \to \infty$.
\begin{lemma}
	\label{Lem_usr_detect}
	If $k_n \log \el_n = o(n)$, then there exist signatures $\bs_i, i=1, \ldots, \el_n$ with  $n''$ channel uses  and energy $E_n''$ such that $P(\cD)$ vanishes as $n \to \infty$.
\end{lemma}
\begin{proof}
	\if	\ISIT 1
	The proof follows along similar lines as that of~\cite[Theorem~2]{ChenCG17}. For details, see  the extended version of this paper~\cite{RaviKISIT20}.
	\else
	The proof follows along similar lines as that of~\cite[Theorem~2]{ChenCG17}. For details, see Appendix~\ref{Sec_Lem_detct_proof}.	
	\fi
\end{proof}

We next use the following lemma to show that  $P\bigl(\cE_m(k'_n)\bigr)$ vanishes as $n \to \infty$ uniformly in  $k'_n \in \cK_n$, where $\cK_n \triangleq \{1, \ldots, \xi k_n\}$.
\begin{lemma}
	\label{Lem_err_expnt}
	Let  $A_{k'_n} \triangleq \frac{1}{k_n'} \sum_{i=1}^{k_n'} \I{ \hat{W}_i \neq W_i}$ and \mbox{$\cA_{k'_n} \triangleq \{1/k_n', \ldots,1 \}$}, where $\I{\cdot}$ denotes the indicator function. Then	for any arbitrary $0<\rho \leq 1$, we have
	\begin{equation}
	\textnormal{Pr}\{A_{k'_n} = a\} \leq \left(\frac{1}{\mu}\right)^{2k'_n} {k_n' \choose a k_n'} M_n^{a k_n' \rho} e^{-nE_0(a, \rho)}, \quad a\in\cA_{k'_n}\label{Eq__random_prob_err}
	\end{equation}
	where
	\begin{align}
	E_0(a, \rho)  \triangleq  \frac{\rho}{2} \ln \left(1+\frac{a 2k_n' E_n'}{n'(\rho +1)N_0}\right) \label{Eq_random_expnt}
	\end{align}
	and
	\begin{align}
	\mu & \triangleq \int \I{ \|\bar{a}\|^2 \leq E_n'} \prod_{i=1}^{n} \tilde{q}(a_i) d \bar{a} \label{Eq_def_mu}
	\end{align}
	is a normalizing constant. In~\eqref{Eq_def_mu}, $\tilde{q}$ denotes the probability density function of a zero-mean Gaussian random variable with variance $E_n'/(2n')$.
\end{lemma}
\begin{proof}
		The upper bound in~\eqref{Eq__random_prob_err}  without the factor $(1/\mu)^{2k'_n}$ can be obtained using random coding with i.i.d. Gaussian inputs~\cite[Theorem~2]{Gallager85}. However, 	while i.i.d. Gaussian codebooks satisfy the energy constraint on average (averaged over all codewords), there may be some codewords in the codebook that violate it. 
		We therefore need to adapt the proof of~\cite[Theorem~2]{Gallager85} as follows.	Let 
		\begin{align*}
		\tilde{\bq}(\bar{a}) & = \prod_{i=1}^{n} \tilde{q}(a_i), \quad \bar{a} =(a_1,\ldots,a_n).
		\end{align*}
	For codewords distributed according to $\tilde{\bq}(\cdot)$, the probability $\Pr(A_{k'_n}=a)$ can be upper-bounded as~\cite[Theorem~2]{Gallager85}
		 \begin{align}
		\Pr(A_{k'_n} = a) & \leq {k_n' \choose a k_n'} M_n^{a k_n' \rho}  \int \tilde{\bq}(\tilde{\bx}_{ak'_n+1}) \cdots \tilde{\bq}(\tilde{\bx}_{k'_n}) \; G ^{1+\rho} \; d\tilde{\bx}_{ak'_n+1} \cdots d\tilde{\bx}_{k'_n} \;d\tilde{\by}  \label{eq_a1}
		\end{align}
		where 
		\begin{align*}
		G & =  \int \tilde{\bq}(x_{1}) \cdots \tilde{\bq}(x_{ak'_n}) \left(  p(\tilde{\by} \mid \tilde{\bx}_{1},\cdots, \tilde{\bx}_{k'_n})\right) ^{1/1+\rho} d\tilde{\bx}_{1} \cdots d\tilde{\bx}_{ak'_n} \notag.
		\end{align*}	
		Using the fact that the channel is memoryless, the RHS of~\eqref{Eq__random_prob_err} without the factor $(1/\mu)^{2k'_n}$ follows from~\eqref{eq_a1}. The  case of $k'_n =2$ was analyzed  in~\cite[Eq.~(2.33)]{Gallager85}.

	Now suppose that all codewords are generated according to the distribution
		\begin{align*}
		\bq(\bar{a}) & = \frac{1}{\mu} \I{ \|\bar{a}\|^2 \leq E_n'} \tilde{\bq}(\bar{a}).
		\end{align*}
		Clearly, such codewords  satisfy the energy constraint $E'_n$ with probability one. Furthermore,
		\begin{align}
		\bq(\bar{a}) & \leq \frac{1}{\mu} \tilde{\bq}(\bar{a}). \label{Eq_prob_signt_uppr1}
		\end{align}
		By replacing $\tilde{\bq}(\cdot)$ in~\eqref{eq_a1} by $\bq(\cdot)$ and upper-bounding $\bq(\cdot)$ by~\eqref{Eq_prob_signt_uppr1}, we obtain that
		\begin{align}
		\textnormal{Pr}\{A_{k'_n} = a\} \leq 		\left(\frac{1}{\mu}\right)^{(1+\rho)(ak'_n)}\left(\frac{1}{\mu}\right)^{k'_n - ak'_n} {k_n' \choose a k_n'} M_n^{a k_n' \rho} e^{-nE_0(a, \rho)}, \quad a\in\cA_{k'_n}. \label{Eq_Prob_An_uppr}
		\end{align}
		From the definition of $\mu$, we have that $0 \leq \mu \leq 1$.
	Since we further have $\rho\leq 1$ and $a\leq 1$, it follows that 
$(1/\mu)^{(1+\rho)(ak'_n)} \leq 	(1/\mu)^{a k_n'+k_n'}$.
Consequently, \eqref{Eq__random_prob_err} follows from~\eqref{Eq_Prob_An_uppr}.
\end{proof}

Next we show that   $\left(\frac{1}{\mu}\right)^{2k'_n} \to 1$ as $n \to \infty$ uniformly in $k'_n \in \cK_n$. By the definition of $\mu$, we have
\begin{align}
\mu & = 1 - \Pr\left(\|\tilde{\bX}_1\|_2^2 \geq E_n'\right) \notag
\end{align}
so $(1/\mu)^{2 k'_n} \geq 1$.
Let us consider $\tilde{\bX}_0 \triangleq \frac{2 n'}{E_n'} \|\tilde{\bX}_1\|_2^2$. Then, 
\begin{align*}
\Pr\left(\|\tilde{\bX}_1\|_2^2 \geq E_n'\right) & = \Pr(\tilde{\bX}_0 \geq 2 n').
\end{align*}
Furthermore, $\tilde{\bX}_0$  has a central chi-square distribution with $n'$ degrees of freedom. So, from
the Chernoff bound we obtain that 
\begin{align*}
\Pr(\tilde{\bX}_0 \geq a) & \leq \frac{E(e^{t\tilde{\bX}_0})}{e^{ta}} \\
& = \frac{(1-2t)^{-n'/2}}{e^{ta}}
\end{align*}
for every $t > 0$. By choosing $a= 2 n'$ and $t= \frac{1}{4}$, this yields
\begin{align}
\Pr(\tilde{\bX}_0 \geq 2 n') & \leq \frac{ \left(\frac{1}{2}\right)^{-n'/2}}{ \exp(n'/2) } \notag \\
& =  \exp \left[-\frac{n'}{2} \tau \right] \notag
\end{align}
where $\tau \triangleq \left( 1 - \ln 2    \right)$ is strictly positive. Thus,
\begin{align}
1 & \leq \left(\frac{1}{\mu}\right)^{2k'_n} \notag \\
& \leq \left(\frac{1}{\mu}\right)^{2\xi k_n} \notag \\
&= (1-\Pr(\tilde{\bX}_0 \geq 2 n'))^{-(2\xi k_n)}\notag \\
& \leq \left(1 - \exp \left[-\frac{n'}{2} \tau \right]\right)^{-(2\xi k_n)}, \quad k'_n \in \cK_n.	 \label{Eq_mu_uppr2}
\end{align}
We have  that $k_n=o(n)$  and $n' = \Theta(n)$. 
Since for any two non-negative sequences $a_n$ and $b_n$ such that $a_n\to 0$ and $a_nb_n \to 0$ as $n \to \infty$, it holds that $(1-a_n)^{-b_n} \to 1$ as $n \to \infty$, we obtain that the RHS of~\eqref{Eq_mu_uppr2} tends to one as $n \to \infty$ uniformly in $k'_n\in \cK_n$. So there exists a positive constant $n_0$  that is independent of $k'_n$ and satisfies
\begin{align*}
\left(\frac{1}{\mu}\right)^{2k'_n} \leq 2, \quad k'_n \in \cK_n, n \geq n_0.
\end{align*}

The probability of error $P\bigl(\cE_m(k'_n)\bigr)$ can be written as
\begin{equation}
P\bigl(\cE_m(k'_n)\bigr) = \sum\limits_{a\in\cA_{k'_n}} \textnormal{Pr}\{A_{k'_n} = a\}.	\label{Eq_prob_err_def}
\end{equation}
So, from Lemma~\ref{Lem_err_expnt}, we obtain
\begin{align}
\textnormal{Pr}\{A_{k'_n} = a\} & \leq 2	{k_n' \choose a k_n'} M_n^{a k_n' \rho}  \exp[-n'E_0(a, \rho)] \notag\\
& \leq 2 \exp\left[ k_n'H_2(a) + a \rho k_n' \log M_n -   n'E_0(a, \rho) \right]\notag \\
& = 2 \exp \left[-E_n'f_{k'_n}(a, \rho)\right], \quad n \geq n_0
\end{align}
where 
\begin{align}
f_{k'_n}(a, \rho) \triangleq \frac{n'E_0(a, \rho)}{E_n'} - \frac{a \rho k_n' \log M_n}{E_n'} - \frac{k_n' H_2(a)}{E_n'}. \label{Eq_fn_def}
\end{align}
We next show that, for sufficiently large $n$, we have
\begin{align}
\textnormal{Pr}\{A_{k'_n} = a\} \leq 2 \exp \left[-E_n'f_{\xi k_n}(1/(\xi k_n), \rho)\right], \quad a\in\cA_{k'_n}, k'_n \in \cK_n. \label{Eq_err_upp_bnd}
\end{align}
To this end, we first note that using basic algebra, we obtain
\begin{align*}
 \frac{d f_{k'_n}(a, \rho)}{da} & \geq \rho k'_n \left[ \frac{1}{1+\frac{2k'_nE'_n}{n'(\rho+1)N_0}} \frac{1}{(1+\rho)N_0} - \frac{\CR}{(1-b)\log e}  \right]\\
 & \geq \rho \left[ \frac{1}{1+\frac{2 \xi k_nE'_n}{n'(\rho+1)N_0}} \frac{1}{(1+\rho)N_0} - \frac{\CR}{(1-b)\log e}  \right].
\end{align*}
This implies that
for any fixed value of $\rho$ and our choices of $E_n'$ and \mbox{$\CR = \frac{(1-b)\log e}{(1+\rho)N_0} - \delta$} (for some arbitrary $0<\delta<\frac{(1-b)\log e}{(1+\rho)N_0}$),
\begin{equation*}
\liminf_{n\to\infty} \min_{k'_n \in \cK_n}  \min_{a \in \cA_{k'_n}} \frac{d f_{k'_n}(a, \rho)}{da} > 0.
\end{equation*}
This follows from the fact that $\frac{k_n E_n'}{n'} \to 0$ as $n \to \infty$, which in turn follows from our choice of $E_n'$ and since  $k_n = o(n / \log n)$. 
So there exists a positive constant $n'_0$ that is independent of $k'_n$ and satisfies
\begin{equation*}
\min_{a \in \cA_{k'_n}}  f_{k'_n}(a, \rho) \geq  f_{ k'_n}(1/k'_n, \rho), \quad k'_n \in \cK_n, n \geq n'_0
\end{equation*}
Furthermore, from the definition of $f_{k'_n}(a, \rho) $ in~\eqref{Eq_fn_def}, it follows  that for $a = 1/k'_n$ and for a given $\rho$, $f_{ k'_n}(a, \rho)$ is decreasing in $k'_n$ since in this case the first two terms on the RHS of~\eqref{Eq_fn_def} are independent of $k'_n$ and the third term is increasing in $k'_n$.
Hence, we can further lower-bound
\begin{equation*}
\min_{a \in \cA_{k'_n}}  f_{k'_n}(a, \rho) \geq  f_{\xi k_n}(1/(\xi k_n), \rho), \quad k'_n \in \cK_n, n \geq n'_0.
\end{equation*}

Next we show that, for our choice of $E_n'$ and \mbox{$\CR$}, we have
\begin{equation}
\label{eq:lim_pos}
\liminf_{n \rightarrow \infty}  f_{\xi k_n}(1/(\xi k_n), \rho) >0.
\end{equation}

Let
\begin{IEEEeqnarray*}{rCl}
	i_n(1/(\xi k_n),\rho) & \triangleq & \frac{n' E_0(1/(\xi k_n),\rho)}{E_n'}\\
	j(\rho) & \triangleq & \frac{\rho \CR}{(1-b)\log e}\\
	h_n(1/(\xi k_n)) & \triangleq & \frac{\xi 	k_n H_2(1/(\xi k_n))}{E_n'}.
\end{IEEEeqnarray*}
Note that $\frac{h_n(1/(\xi k_n))}{j(\rho)}$ vanishes as $n \to \infty$ for our choice of $E_n'$.
Consequently,
\begin{IEEEeqnarray*}{lCl}
	\liminf_{n \rightarrow \infty} f_n(1/(\xi k_n), \rho) 
	& = & j(\rho)  \biggl\{\liminf_{n\to\infty} \frac{i_n(1/(\xi k_n),\rho)}{j(\rho)} - 1 \biggr\}.
\end{IEEEeqnarray*}
The term  $j( \rho) =\rho \CR/(1-b)\log e$ is bounded away from zero for our choice of $\CR$ and $\delta < \frac{(1-b)\log e}{(1+\rho)N_0}$.  Furthermore, since $E_n'/n' \to 0$, we get
\begin{equation*}
\lim_{n\to\infty} \frac{i_n(1/(\xi k_n),\rho)}{j(\rho)} = \frac{(1-b)\log e}{(1+\rho)N_0 \CR}
\end{equation*}
which is strictly larger than $1$ for our choice of $\CR$. So, \eqref{eq:lim_pos} follows. Consequently, there exist two positive constants $\gamma$ and $n''_0$ that are independent of $k'_n$ and satisfy $f_{k'_n}(a, \rho)  \geq \gamma $ for $a\in \cA_{k'_n}$, $k'_n\in \cK_n$, and  $n \geq n''_0$.
We conclude that  for $n \geq \max (n_0,n'_0,n''_0)$,
\begin{align}
\textnormal{Pr}\{A_{k'_n} = a\} \leq 2 e^{-E_n'\gamma}, \quad a \in \cA_{k'_n}, k'_n \in \cK_n. \label{Eq_type_uppr}
\end{align}
Since $|\cA_{k'_n}| = k_n'$, it follows from~\eqref{Eq_prob_err_def} and~\eqref{Eq_type_uppr} that
\begin{align}
P\bigl(\cE_m(k'_n)\bigr) \leq k_n' 2 e^{-E_n'\gamma}, \quad k'_n \in \cK_n, n \geq \max (n_0,n'_0,n''_0). \notag
\end{align}
Further upper-bounding $k'_n \leq \xi k_n$, this implies that
\begin{align}
\sum_{k'_n=1}^{\xi k_n}\Pr\{K'_n=k_n'\}P\bigl(\cE_m(k'_n)\bigr) & \leq \xi k_n 2 e^{-E_n'\gamma},  \quad n \geq \max (n_0,n'_0,n''_0).  \label{Eq_sum_prob_uppr}
\end{align}
Since $E_n' = (1-b)c_n \ln \el_n$ and  $k_n = O(\el_n)$, it follows that the RHS of~\eqref{Eq_sum_prob_uppr} tends to 0 as $n \to \infty$ for our choice of \mbox{$\CR = \frac{(1-b)\log e}{(1+\rho)N_0} - \delta$}.
Since $\rho,\delta,$ and $b$ are  arbitrary, any rate $\CR < \frac{\log e}{N_0}$ is thus achievable. This proves Part~\ref{Thm_achv_part}) of Theorem~\ref{Thm_random_JPE}.

Next we prove Part~\ref{Thm_conv_part}).
Let $\hat{W_i}$ denote the receiver's estimate of $W_i$, and
denote by $\bW$ and $\bhW$ the vectors $(W_1,\ldots, W_{\el_n})$ and $(\hat{W_1},\ldots,\hat{W}_{\el_n})$, respectively. 
The messages $W_1,\ldots, W_{\el_n}$ are independent, so it follows from~\eqref{Eq_messge_def} that
\begin{align*}
H(\bW) = \el_n H(\bW_1) = \el_n \left(H_2(\alpha_n) + \alpha_n \log M_n \right)
\end{align*} 
where $H_2(\cdot)$ denotes the binary entropy function.
Since $	H(\bW) = H(\bW|\bY)+I(\bW;\bY)$, we obtain
\begin{align}
\el_n \left(H_2(\alpha_n) + \alpha_n \log M_n \right) & =H(\bW|\bY)+I(\bW;\bY). \label{Eq_messge_entrpy}
\end{align}
To bound $H(\bW)$, we  use the upper bounds~\cite[Lemma~2]{ChenCG17}
\begin{align}
H(\bW|\bY) \leq & \log 4 + 4 P_{e}^{(n)}\big(k_n \log M_n + k_n  + \el_n H_2(\alpha_n) + \log M_n \big) \label{Eq_messg_cond_entrpy}
\end{align}
and~\cite[Lemma~1]{ChenCG17}  
\begin{align}
I(\bW;\bY)  \leq  \frac{n}{2} \log \left(1+\frac{ 2k_nE_n}{nf_{k'_n}}\right). \label{Eq_mutl_info_uppr}
\end{align}
Using~\eqref{Eq_messg_cond_entrpy} and~\eqref{Eq_mutl_info_uppr}  in~\eqref{Eq_messge_entrpy}, rearranging terms, and dividing by $k_nE_n$, yields
\begin{align}
\left(1-4 P_{e}^{(n)}(1+1/k_n)\right) \CR \leq & \frac{\log 4}{k_nE_n} +   \frac{H_2(\alpha_n)}{\alpha_n E_n} \! \left(4 P_{e}^{(n)} -1\right) \notag \\
& \quad   +  4 P_{e}^{(n)} (1/E_n + 1/k_n) +\frac{n}{2 k_nE_n} \log \left(1+\frac{ 2k_nE_n}{nN_0}\right)\! . \label{Eq_rate_joint_uppr}
\end{align}
We next show that if  $k_n \log \el_n = \omega(n)$, then the right-hand side (RHS) of~\eqref{Eq_rate_joint_uppr} tends to a non-positive value. To this end, we need the following lemma.
\begin{lemma}
	\label{Lem_energy_bound}
	If $\CR > 0$, then $P_{e}^{(n)}$ vanishes as $n\to\infty$ only if  \mbox{$E_n = \Omega(\log \el_n)$}.
\end{lemma}
\begin{proof}
	\if \ISIT 1
	The proof of this lemma follows along similar lines as that of~\cite[Lemma~2]{RaviKISIT19}. For details, see~\cite{RaviKISIT20}.
	\else	
	See Appendix~\ref{Append_prob_lemma}.
	\fi
\end{proof}

Part~\ref{Thm_conv_part}) of Theorem~\ref{Thm_random_JPE} follows now by contradiction. Indeed, let us assume that $k_n \log \el_n = \omega(n)$, $P_{e}^{(n)} \to 0$,  and $\CR >0$. Then, Lemma~\ref{Lem_energy_bound} together with the assumption that $k_n = \Omega(1)$ implies that  $k_nE_n = \omega(n)$. It follows that the last term on the RHS of~\eqref{Eq_rate_joint_uppr} tends to zero as $n \to \infty$.  The assumption  $k_n \log \el_n = \omega(n)$ in turn implies that $\el_n \to \infty$ as $n \to \infty$. So, by Lemma~\ref{Lem_energy_bound}, $E_n \to \infty$.
Together with the assumption that $k_n = \Omega(1)$, this implies that the first and third term on the RHS of~\eqref{Eq_rate_joint_uppr} vanish as $n \to \infty$. Finally, 
$\frac{H_2(\alpha_n)}{\alpha_n E_n}$ is a sequence of non-negative numbers and $(4 P_{e}^{(n)} -1) \to -1$ as $n \to \infty$, so the second term converges to a non-positive  value. Thus, we obtain that $\CR$ tends to a non-positive value as $n \to \infty$. 
This contradicts the assumption $\CR > 0$, so Part~\ref{Thm_conv_part}) of Theorem~\ref{Thm_random_JPE} follows.

\subsection{Proof of Theorem~\ref{Thm_ortho_accs}}
\label{Sec_proof_ortho_access}

To prove Part~\ref{Thm_ortho_accs_achv}), we present a scheme that is similar to the one given in~\cite{RaviKISIT19} for the non-random-access case.  Specifically, each user is assigned  $n/\el_n$ channel uses out of which the first one is used for sending a pilot signal  and the rest are used for sending the message. Out of the  available energy $E_n$,   $t E_n$ for some arbitrary $0 < t < 1$ is used for the pilot signal and $(1-t)E_n$ is used for sending the message. Let $\tilde{\bx}(w) $ denote the codeword of length $\frac{n}{\el_n}-1$ for sending message $w$. Then
user $i$ sends in his assigned slot the codeword 
\begin{align*}
\bx(w_i) = \left(\sqrt{t E_n}, \tilde{\bx}(w_i)\right).
\end{align*}
The receiver first detects from the pilot signal whether user $i$ is active or not. If the user is estimated as active, then it decodes the user's message.
Let  $P_i = \textnormal{Pr}\{\hat{W_i} \neq W_i\}$ denote the probability that user $i$'s message is decoded erroneously.
Since all users follow the same coding scheme, the probability of correct decoding is given by
\begin{align}
P_c^{(n)}  = \left(1-P_1\right)^{\el_n}. \label{Eq_ortho_corrct}
\end{align}
By employing the transmission scheme that was used to prove~\cite[Theorem~2]{RaviKISIT19}, we get an upper bound on the probability of error $P_1$ as follows.
\begin{lemma}
	\label{Lem_detect_uppr}
	For $n\geq n_0$ and sufficiently large $n_0$, the probability of error in decoding  user 1's message can be upper-bounded as: 
	\begin{align*}
	P_1 & \leq \frac{2}{n^2}.
	\end{align*}
\end{lemma}
\begin{proof}
	\if \ISIT 1
	See~\cite{RaviKISIT20}.
	\else
	See Appendix~\ref{Sec_appnd_ortho}.
	\fi
\end{proof}

From Lemma~\ref{Lem_detect_uppr} and~\eqref{Eq_ortho_corrct},
\begin{align}
P_c^{(n)} &  \geq \left(1-\frac{2}{n^{2}}\right)^{\el_n} \notag \\
& \geq	\left(1-\frac{2}{n^{2}}\right)^{\frac{n}{\log n}} \notag
\end{align}
which tends to one as $n \to \infty$. Thus,  Part~\ref{Thm_ortho_accs_achv}) of Theorem~\ref{Thm_ortho_accs} follows.

To prove Part~\ref{Thm_ortho_accs_conv}), we first note  that we consider symmetric codes, i.e., the pair $(M_n,E_n)$ is the same for all users. However, each user may be assigned different numbers of channel uses. Let $n_i$ denote the number of channel uses assigned to  user $i$. For an orthogonal-access scheme, if $\el_n = \omega(n/ \log n)$, then there exists at least one user, say $i=1$, such that $n_i = o(\log n)$.  
Using that $H(W_1 |  W_1 \neq 0 ) = \log M_n$, it follows from Fano's inequality that
\begin{align}
\log M_n	& \leq 1+P_1 \log M_n + \frac{n_1}{2 }\log\left(1+\frac{ 2 E_n}{n_1N_0}\right).  \nonumber 
\end{align}
This implies that the rate per unit-energy $\CR=(\log M_n)/E_n$ for user 1 is upper-bounded by 
\begin{align}
\CR \leq \frac{ \frac{1}{E_n} + \frac{n_1}{2 E_n}\log\left(1+\frac{ 2E_n}{n_1N_0}\right)}{1 -P_1}.\label{Eq_R_avg}
\end{align}
Since $\el_n = \omega(n/ \log n)$, it follows from Lemma~\ref{Lem_energy_bound} that $P_{e}^{(n)}$ goes to zero only if 
\begin{align}
E_n = \Omega(\log n). \label{Eq_ortho_enrg_lowr}
\end{align}
In contrast, \eqref{Eq_R_avg}  implies  that $\CR>0$ only if 
$E_n = O(n_1)$. Since $n_1 = o(\log n)$, this further implies that 
\begin{align}
E_n = o(\log n). \label{Eq_ortho_enrg_uppr}
\end{align}
No sequence $\{E_n\}$ can satisfy both~\eqref{Eq_ortho_enrg_uppr} and~\eqref{Eq_ortho_enrg_lowr} simultaneously. We thus obtain that if $\el_n =\omega(n/ \log n)$, then the capacity per unit-energy is zero. This is Part~\ref{Thm_ortho_accs_conv}) of Theorem~\ref{Thm_ortho_accs}.

	\label{sec_average}
Many works in the literature on many-access channels, including \cite{Polyanskiy17, OrdentlichP17,VemNCC17,ZadikPT19,KowshikPISIT19,KowshikP19,KowshiKAFPISIT19,KowshiKAFP19}, consider a \emph{per-user probability of error}
\begin{equation}
\label{eq:Pe_A}
P_{e,A}^{(n)} \triangleq \frac{1}{\el_n} \sum_{i=1}^{\el_n} \textnormal{Pr}\{\hat{W_i} \neq W_i\}
\end{equation}
rather than the joint error probability \eqref{Eq_prob_err}. In the following, we briefly discuss the behavior of the capacity per unit-energy when the error probability is $P_{e,A}^{(n)}$, which in this paper we shall refer to as \emph{average probability of error (APE)}. To this end, we define an $(n,\{M_n^{(\cdot)}\},\{E_n^{(\cdot)}\}, \epsilon)$ code under APE with the same encoding and decoding functions defined in Section~\ref{Sec_model}, but with the probability of error \eqref{Eq_prob_err} replaced with \eqref{eq:Pe_A}. We denote the capacity per unit-energy under APE by $\CC^A$.

Under APE, if $\alpha_n \to 0$ as $n \to \infty$, then $\Pr\{W_i =0\} \to 1$ for all $i=1,\ldots,\el_n$.
Consequently, a code with $M_n =2$ and $E_n =0$ for all $n$ and a decoding function that always declares that all users are inactive achieves an APE that vanishes as $n \to \infty$.
This implies that $\CC^A = \infty$ for vanishing $\alpha_n$. 
In the following, we avoid this trivial case and assume that $\alpha_n$ is bounded away from zero.

For a Gaussian MnAC with APE and $\alpha_n=1$ (non-random-access case),  we showed in~\cite{RaviKIZS20} that if the number of users grows sublinear in $n$, then each user can achieve the single-user capacity per unit-energy, and if the order of growth is linear or superlinear, then the capacity per unit-energy is zero. Perhaps not surprisingly, the same result holds in the random-access case since, when $\alpha_n$ is bounded away from zero, $k_n$ is of the same order as $\el_n$.
\begin{theorem}
	\label{Thm_capac_PUPE}
	If $k_n = \Theta( \el_n)$ and $\alpha_n \to \alpha \in (0,1]$, then $\CC^A$ has the following behavior:
	\begin{enumerate}
		\item  If  $\el_n  = o(n)$, then $\CC^A = \frac{\log e}{N_0}$. Moreover, the capacity per unit-energy can be achieved by an orthogonal-access scheme where each user uses a codebook with orthogonal codewords. \label{Thm__avg_achv_part}
		\item If $\el_n  = \Omega(n)$, then $\CC^A =0$.	\label{Thm__avg_conv_part}
	\end{enumerate}
\end{theorem}
\begin{proof}
		To prove Part~\ref{Thm__avg_achv_part}), we first argue that $P_{e,A}^{(n)} \to 0$ only if  $E_n \to \infty$. Indeed, we have
	\begin{align*}
	P_{e,A}^{(n)} & \geq \min_{i} \text{Pr}\{\hat{W}_i\neq W_i\} \\
	& \geq \alpha_n \text{Pr}(\hat{W_i} \neq W_i | W_i \neq 0) \; \text{ for some } i.
	\end{align*}
	Since $\alpha_n\to\alpha > 0$, this implies that 
	$P_{e,A}^{(n)} $ vanishes only if $ \text{Pr}(\hat{W_i} \neq W_i | W_i \neq 0)$ vanishes.  We next note that \mbox{$\text{Pr}(\hat{W_i} \neq W_i | W_i \neq 0)$} is lower-bounded by the error probability of the Gaussian single-user channel. By following the arguments in the proof of \cite[Theorem~2, Part~1)]{RaviKIZS20}, we obtain that  $P_{e,A}^{(n)} \to 0$ only if $E_n \to \infty$, which also implies that $\CC^A \leq \frac{\log e}{N_0}$.
	
	We next show that any rate per unit-energy $\CR < \frac{\log e}{N_0}$ is achievable by an orthogonal-access scheme where each user uses an orthogonal codebook of blocklength $n/\el_n$. Out of these $n/\el_n$ channel uses, the first one is used for sending a pilot signal to convey whether the user is active or not, and the remaining channel uses are used to send the message. Specifically, to transmit message $w_i$, user $i$ sends in his assigned slot the codeword $\bx(w_i) = (x_1(w_1), \ldots, x_{n/\el_n }(w_i))$, which is given by 
	\begin{align*}
	x_{k}(w_i) = \begin{cases}
	\sqrt{t E_n}, & \text{ if } k=1  \\
	\sqrt{(1-t) E_n}, & \text{ if }  k=w_i+1\\
	0, & \text{ otherwise}.
	\end{cases}
	\end{align*}
	\if \ISIT 1
	From the pilot signal, the receiver first detects whether the user is active or not. As shown in the proof of Lemma~\ref{Lem_detect_uppr}, the detection error vanishes as $n \to \infty$. Using the upper bound on the decoding-error probability for an orthogonal code with $M$ codewords and rate per unit-energy $\CR$ given in~\cite[Lemma~3]{RaviK19}, we can then show that $P_i, i=1,\ldots, \el_n$ vanishes as $n$ tends to infinity. This implies that also $P_{e,A}^{(n)} $ vanishes as $n \to \infty$. More details can be found in~\cite{RaviKISIT20}.
	\else 
	From the pilot signal, the receiver first detects whether the user is active or not. As shown in the proof of Lemma~\ref{Lem_detect_uppr} that the detection error vanishes as $n \to \infty$.	
	Furthermore, the probability of error in decoding for an orthogonal code with $M$ codewords and rate per unit-energy $\CR$ for the AWGN channel is upper-bounded by~\cite[Lemma~3]{RaviK19}:
	\begin{align}
	P_e \leq
	\begin{cases}
	\exp\left\{- \frac{\ln M}{\CR}\left(\frac{\log e}{2 N_0} - \CR \right) \right\}, \text{ if } 0 < \CR \leq \frac{1}{4} \frac{\log e}{N_0}\\
	\exp\left\{- \frac{\ln M}{\CR}  \left(\sqrt{\frac{\log e}{N_0}} - \sqrt{ \CR}\right)^2\right\}, \text{ if } \frac{1}{4} \frac{\log e}{N_0}  \leq \CR \leq \frac{\log e}{N_0}.
	\end{cases}
	\label{Eq_ortho_prob_uppr}
	\end{align}
	It follows from~\eqref{Eq_ortho_prob_uppr} that if $\CR < \frac{\log e}{N_0}$ and $M \to \infty$ as $n \to \infty$, then $P_e$ tends to zero as $n \to \infty$. Since $\el_n =o(n)$, it follows that $M = n/\el_n -1$ tends to $\infty$, as $n \to \infty$. Thus, for any $\CR < \frac{\log e}{N_0}$, the  probability of error in decoding vanishes. Thus, we obtain that $P_i, i=1,\ldots, \el_n,$ vanishes as $n$ tends to infinity
	This implies that also $P_{e,A}^{(n)} $ vanishes as $n \to \infty$.
	\fi
	
	\if \ISIT 1
	The proof of Part~\ref{Thm__avg_conv_part}) follows from Fano's inequality and is similar to that of~\cite[Theorem~2, Part~2)]{RaviKIZS20}. Details can be found in~\cite{RaviKISIT20}.
	\else
	Now we prove Part~\ref{Thm__avg_conv_part}). Fano's inequality yields that $H(\hat{W}_i|W_i) \leq 1+ P_i\log M_n$. Since $H(W_i) = H_2(\alpha_n) + \alpha_n \log M_n$, we have
	\begin{equation*}
	H_2(\alpha_n) + \alpha_n \log M_n \leq 1+ P_i\log M_n+ I(W_i; \hat{W}_i)
	\end{equation*}
	for $i=1,\ldots, \el_n$. Averaging over all $i$'s then gives
	\begin{align}
	H_2(\alpha_n) + \alpha_n \log M_n &  \leq  1+  \frac{1}{\el_n} \sum_{i=1}^{\el_n} P_i\log M_n+ \frac{1}{\el_n}  I({\bf W}; {\bf \hat{W}}) \nonumber\\ 
	&  \leq   1+P_{e,A}^{(n)}\log M_n+ \frac{1}{\el_n} I(\bW; \bY) \nonumber\\
	&  \leq 1 +  P_{e,A}^{(n)} \log M_n+\frac{n}{2\el_n} \log \left(1+\frac{2 k_nE_n}{nN_0}\right). \label{Eq_avg_prob_uppr}
	\end{align}
	Here, the first inequality follows because the messages $W_i, i=1, \ldots, \el_n$ are independent and because conditioning reduces entropy, the second inequality follows from the definition of $P_{e,A}^{(n)}$ and the data processing inequality, and the third inequality follows by upper-bounding $I(\bW;\bY)$ by $\frac{n}{2} \log \bigl(1+\frac{2 k_nE_n}{nN_0}\bigr)$ \cite[Lemma~1]{ChenCG17}.
	
	Dividing both sides of \eqref{Eq_avg_prob_uppr} by $E_n$, and rearranging terms,
	yields an upper-bound on the rate per unit-energy $\CR^A$  as
	\begin{equation}
	\label{eq:Part2_Th1_end}
	\PR\leq \frac{ \frac{1 - H_2(\alpha_n)}{E_n} + \frac{n}{2 \el_nE_n}\log(1+\frac{ 2k_nE_n}{nN_0})}{\alpha_n -P_{e,A}^{(n)}}.
	\end{equation}
	As noted before, $P_{e,A}^{(n)} \to 0$ only if $E_n \to \infty$.
	It follows that  $\frac{1 - H_2(\alpha_n)}{E_n}$ vanishes as $n \to \infty$.
	Furthermore, together with the assumptions $\el_n=\Omega(n)$ and $\alpha_n \to \alpha>0$, $E_n\to\infty$ yields that $k_nE_n/n=\ell_nE_n/(\alpha_n E_n)$ tends to infinity as $n\to\infty$. This in turn implies that
	\begin{equation*}
	\frac{n}{2\ell_n E_n} \log\left(1+\frac{2 k_n E_n}{n N_0}\right)=\frac{n \alpha_n}{2 k_n E_n}\log\left(1+\frac{2k_n E_n}{n N_0}\right)
	\end{equation*}
	vanishes as $n\to \infty $. It thus follows from~\eqref{eq:Part2_Th1_end} that $\CR^A$ vanishes as $n\to\infty$, thereby proving  Part~\ref{Thm__avg_conv_part}) of Theorem~\ref{Thm_capac_PUPE}.
\end{proof}

\if \ISIT 0

\appendices
\input{appendix.tex}

\fi



\end{document}

%% file: appendix.tex
\section{Proof of Lemma~\ref{Lem_usr_detect}}
\label{Sec_Lem_detct_proof}

First let us consider the case of bounded $\el_n$.
In this case, one can employ a scheme where each user gets an exclusive channel use to convey whether it is active or not. For such a scheme, it is easy to show that (see the proof of Lemma~\ref{Lem_detect_uppr} in Appendix~\ref{Sec_appnd_ortho}) the probability of a detection error $P(\cD)$ is upper-bounded by
\begin{align*}
P(\cD) & \leq \el_n e^{-E_n'' t}
\end{align*}
for some $t > 0$. The energy $E_n'' = b c_n \ln \el_n$ used for detection  tends to infinity  since $c_n \to \infty$ as $n \to\infty$. 
Thus,  $P(\cD)$ tends to zero as $n \to \infty$.

Next we prove Lemma~\ref{Lem_usr_detect} for the case where $\el_n \to  \infty$ as $n \to \infty$. To this end, we closely follow the proof of~\cite[Theorem~2]{ChenCG17}, but with the power constraint replaced by an energy constraint.
Specifically, we analyze  $\Pr(\cD)$  for  the user-detection scheme given in~\cite{ChenCG17} where signatures are drawn i.i.d. according to a zero mean Gaussian distribution.
 Note that the proof in~\cite{ChenCG17} assumes that
\begin{align}
\lim\limits_{n \to \infty} \el_n e^{-\delta k_n} =0 \label{Eq_Guo_cond}
\end{align}
for all $\delta > 0$. However, as we shall show next, in our case this assumption is not necessary.

To show that all signatures satisfy the energy constraint, we follow the technique used in the proof of Lemma~\ref{Lem_err_expnt}. Similar to Lemma~\ref{Lem_err_expnt}, we denote by $\tilde{q}$  the probability density function of a zero-mean Gaussian random variable with variance $E_n''/(2n'')$. We further let
\begin{align*}
\tilde{\bq}(\bar{a}) & = \prod_{i=1}^{n} \tilde{q}(a_i), \quad \bar{a} =(a_1,\ldots,a_n)
\end{align*}
and
\begin{align*}
\bq(\bar{a}) & = \frac{1}{\mu} \I{ \|\bar{a}\|^2 \leq E_n''} \tilde{\bq}(\bar{a}) 
\end{align*}
where
\begin{align*}
\mu & = \int \I{ \|\bar{a}\|^2 \leq E_n''} \; \tilde{\bq}(\bar{a})  d \bar{a}
\end{align*}
is a normalizing constant.  
Clearly, any vector $\bS_i$ distributed according to $\bq(\cdot)$ satisfies the energy constraint $E''_n$ with probability one. 
For any index set $I \subseteq \{1,\ldots, \el_n\}$, let the matrices $\underline{\bS}_{I}$ and $ \tilde{\underline{\bS}}_{I}$ denote the set of signatures for the users in $I$ that are distributed respectively as
\begin{align*}
\underline{\bS}_{I} & \sim \prod_{i\in I} \bq(\bS_i)
\end{align*}
and
\begin{align*}
\tilde{\underline{\bS}}_{I} & \sim \prod_{i\in I} \tilde{\bq}(\bS_i).
\end{align*}
As noted in the proof of Lemma~\ref{Lem_err_expnt}, we have
\begin{align}
\bq(\bs_i) & \leq \frac{1}{\mu} \tilde{\bq}( \bs_i). \label{Eq_prob_signt_uppr}
\end{align}

To analyze the detection error probability, we first define the $\el_n$-length vector $\bD^a$ as 
\begin{align*}
\bD^a \triangleq ( \I{W_1\neq0}, \ldots, \I{W_{\el_n} \neq 0}).
\end{align*}
For $c_n$ given in~\eqref{Eq_energy_choice}, let
\begin{align*}
v_n \triangleq k_n(1 + c_n).
\end{align*}
Further let
\begin{align*}
\cB^n(v_n) \triangleq \{ \bd \in \{0,1 \}^{\el_n} : 1 \leq  |\bd| \leq v_n \}
\end{align*}
where $|\bd|$ denotes the number of $1$'s in $\bd$. We denote by $\bS^a$  the matrix of signatures of all users which are generated independently according to $\bq(\cdot)$, and 
we denote by $\mathbf{Y}^a$ the first $n''$ received symbols, based on which the receiver performs user detection. The receiver outputs the $\hat{\bd}$ given by
\begin{align}
\hat{\bd} = \mathrm{ arg\,min}_{ \bd \in \cB^n(v_n) } \| \bY^a - \bS^a \bd \| \label{Eq_decod_rule}
\end{align}
as a length-$\el_n$ vector indicating the set of active users. Then, the probability of a detection error $\Pr(\cD)$ is upper-bounded by
\begin{align}
\Pr(\cD) 
& \leq \Pr(|\bD^a| > v_n ) + \sum_{\bd \in \cB^n(v_n)} \Pr(\cE_d|\bD^a = \bd) \Pr(\bD^a = \bd) + \Pr(\cE_d| |\bD^a| = 0) \Pr(|\bD^a| = 0) \label{Eq_detect_err_uoor}
\end{align}
where $|\bD^a|$ denotes the number of $1$'s in $\bD^a$ and $\Pr(\cE_d|\bD^a = \bd)$ denotes the detection error probability for a given $\bD^a = \bd$.  Next we show that each term on the RHS of~\eqref{Eq_detect_err_uoor} vanishes as $n \to \infty$. 

Using the Chernoff bound for the binomial distribution, we have
\begin{align*}
\Pr(|\bD^a| > v_n)  & \leq \exp(-k_n c_n/3)
\end{align*}
which vanishes since $c_n \to \infty $ and $k_n = \Omega(1)$.

We continue with the term $\Pr(\cE_d|\bD^a = \bd)$. For a given $\bD^a = \bd$, let $\kappa_1$ and $\kappa_2$ denote the number of miss detections and false alarms, respectively, i.e.,
\begin{align*}
\kappa_1 &= |\{ j: d_j \neq 0, \hat{d}_j = 0 \} |\\
\kappa_2 &= |\{ j: d_j = 0, \hat{d}_j \neq 0 \} |
\end{align*}
where $d_j$ and $\hat{d}_j$ denote the $j$-th components of the corresponding vectors. An error happens only if either $\kappa_1$ or $\kappa_2$ or both are strictly positive. The number of users that are either active or are declared as active by the receiver satisfies $|\bd|+\kappa_2 = |\hat{\bd}|+\kappa_1$, so
\begin{align*}
|\bd|+ \kappa_2 & \leq v_n + \kappa_1
\end{align*}
since $|\hat{\bd}|$ is upper-bounded by $v_n$ by the decoding rule~\eqref{Eq_decod_rule}.
So, the pair $(\kappa_1, \kappa_2)$ belongs to the following set:
\begin{align}
\cW^{\el_n}_{\bd} = & \left\{(\kappa_1,\kappa_2): \kappa_1 \in \{0,1,\ldots, |\bd| \}, \kappa_2 \in \{0,1,\ldots,v_n\},   \kappa_1+\kappa_2 \geq 1, |\bd|+\kappa_2 \leq v_n + \kappa_1 \right\}. \label{Eq_decision_sets}
\end{align}
Let $\Pr(\cE_{\kappa_1,\kappa_2}|\bD^a = \bd)$ be the probability of having  exactly $\kappa_1$  miss detections  and $\kappa_2$ false alarms when $\bD^a = \bd$. 
For  given $ \bd$ and $\hat{\bd}$, let $A^* = \{j : d_j \neq 0 \}$ and $A = \{j : \hat{d}_j \neq 0 \}$. We further define $A_1 = A^* \setminus A$, $A_2 = A \setminus A^*$, and 
\begin{align*}
T_A & = \|\bY^a - \sum_{j \in A} \bS_j \|^2 -  \|\bY^a - \sum_{j \in A^*} \bS_j \|^2.
\end{align*}
Using the analysis that led to~\cite[eq.~(67)]{ChenCG17}, we obtain
\begin{align}
\Pr(\cE_{\kappa_1,\kappa_2}|\bD^a = \bd) & \leq \binom{|A^*|}{\kappa_1} \binom{\el_n}{\kappa_2} \mathrm{E}_{\underline{\bS}_{A^*}, \bY} \{ [\mathrm{E}_{\underline{\bS}_{A_2}}\{ \I{T_A \leq 0} |\underline{\bS}_{A^*}, \bY \}]^{\rho}|\} \notag \\
& \leq \binom{|A^*|}{\kappa_1} \binom{\el_n}{\kappa_2}  \left(\frac{1}{\mu} \right)^{\rho \kappa_2} \mathrm{E}_{\underline{\bS}_{A^*}, \bY} \{ [\mathrm{E}_{\underline{\tilde{\bS}}_{A_2}}\{ \I{T_A \leq 0} |\underline{\bS}_{A^*}, \bY \}]^{\rho}\} \notag\\
& \leq \binom{|A^*|}{\kappa_1} \binom{\el_n}{\kappa_2} \left(\frac{1}{\mu}\right)^{|A^*|} \left(\frac{1}{\mu} \right)^{\rho \kappa_2} \mathrm{E}_{\underline{\tilde{\bS}}_{A^*}, \bY} \{ [\mathrm{E}_{\underline{\tilde{\bS}}_{A_2}}\{ \I{T_A \leq 0} |\underline{\bS}_{A^*}, \bY \}]^{\rho} \} \label{Eq_detect_err_new_distr2}
\end{align}
where in the second inequality we used that 
\begin{align}
\bq( \underline{\bs}_{A_2}) \leq \left( \frac{1}{\mu}\right)^{\kappa_2}\prod_{i \in A_2} \tilde{\bq}(\bs_i ) \label{Eq_Q_uppr1}
\end{align}
and in the third inequality we used that  
\begin{align}
\bq(\underline{\bs}_{A^*}) \leq \left( \frac{1}{\mu}\right)^{|A^*|}\prod_{i \in A^*} \tilde{\bq}(\bs_i ). \label{Eq_Q_uppr2}
\end{align}
Here, \eqref{Eq_Q_uppr1} and~\eqref{Eq_Q_uppr2} follow from~\eqref{Eq_prob_signt_uppr}.

 For every $\rho \in [0,1]$ and $\lambda \geq 0$, we obtain 	 from~\cite[eq.~(78)]{ChenCG17} that
\begin{align}
\binom{|A^*|}{\kappa_1} \binom{\el_n}{\kappa_2} \mathrm{E}_{\underline{\tilde{\bS}}_{A^*}, \bY} \{ [\mathrm{E}_{\underline{\tilde{\bS}}_{A}}\{ \I{T_A \leq 0} \}]^{\rho} \} & \leq \exp[ -\tilde{E}_n g^n_{\lambda, \rho}(\kappa_1,\kappa_2, \bd) ] \label{Eq_detect_prob_old_distrb}
\end{align}
where 
\begin{align}
\tilde{E}_n  \triangleq & E_n''/2, \notag \\
g^n_{\lambda, \rho}(\kappa_1,\kappa_2, \bd)  \triangleq & -\frac{(1-\rho)n''}{2 \tilde{E}_n} \log (1+\lambda \kappa_2 \tilde{E}_n/n'') + \frac{n''}{2\tilde{E}_n} \log \left(1+ \lambda(1-\lambda \rho)\kappa_2 \tilde{E}_n/n'' + \lambda \rho (1-\lambda \rho) \kappa_1 \tilde{E}_n/n''\right) \notag \\
&  - \frac{|\bd|}{\tilde{E}_n} H_2\left(\frac{\kappa_1}{|\bd|}\right) - \frac{\rho \el_n}{\tilde{E}_n} H_2\left(\frac{\kappa_2}{\el_n}\right). \label{Eq_err_exp}
\end{align}
Thus, it follows from~\eqref{Eq_detect_err_new_distr2} and \eqref{Eq_detect_prob_old_distrb} that
\begin{align}
\Pr(\cE_{\kappa_1,\kappa_2}|\bD^a = \bd) & \leq \left(\frac{1}{\mu}\right)^{|A^*| + \rho\kappa_2}  \exp[ -\tilde{E}_n g^n_{\lambda, \rho}(\kappa_1,\kappa_2, \bd) ].\label{Eq_err_mu_uppr}
\end{align}

Next we show that the RHS of~\eqref{Eq_err_mu_uppr} vanishes as $n \to \infty$. To this end,  we first show that $\left(\frac{1}{\mu}\right)^{|A^*| + \rho\kappa_2} \to 1$ as $n \to \infty$ uniformly in $(\kappa_1, \kappa_2) \in \cW^{\el_n}_{\bd}$ and $\bd \in \cB^n(v_n) $. 
From the definition of $\mu$, we have
\begin{align*}
\mu & = 1 - \Pr\left(\|\tilde{\bS}_1\|_2^2 \geq E_n''\right).
\end{align*}
Further, by defining $\tilde{\bS}_0 \triangleq \frac{2 n''}{E_n''} \|\tilde{\bS}_1\|_2^2$ and  following the steps that led to~\eqref{Eq_mu_uppr2}, we obtain
\begin{align}
1 & \leq \left(\frac{1}{\mu}\right)^{|A^*| + \rho\kappa_2} \notag \\
& \leq \left(\frac{1}{\mu}\right)^{2 \el_n} \notag\\
& = (1-\Pr(\tilde{\bS}_0 \geq 2 n''))^{-2 \el_n}\notag \\
& \leq \left(1 - \exp \left[-\frac{n''}{2} \tau \right]\right)^{-2 \el_n} \label{Eq_mu_uppr}
\end{align}
where $\tau = (1 - \ln 2)$. Here, in the second inequality we used that $|A^*| \leq \el_n$ and $\rho \kappa_2 \leq \el_n$.
Since $k_n \log \el_n = o(n)$ and $k_n = \Omega(1)$, we have $\log \el_n =o(n)$. Furthermore, $n'' = \Theta(n)$. 
As noted before, for any two non-negative sequences $a_n$ and $b_n$ satisfying $a_n\to 0$ and $a_nb_n \to 0$ as $n \to \infty$, it holds that $(1-a_n)^{-b_n} \to 1$ as $n \to \infty$. So, we obtain that the RHS of~\eqref{Eq_mu_uppr} goes to one as $n \to \infty$ uniformly in $(\kappa_1, \kappa_2)\in \cW_{\bd}^{\el_n}$ and $\bd \in \cB^n(v_n)$. So there exists a positive constant $n_0$ that is independent of $\kappa_1$, $\kappa_2$, and $\bd$ and satisfies 
\begin{align}
\left( \frac{1}{\mu} \right)^{|A^*| + \rho\kappa_2} & \leq 2, \quad (\kappa_1, \kappa_2)\in \cW_{\bd}^{\el_n}, \bd \in \cB^n(v_n), n \geq n_0. \label{Eq_mu_lowr}
\end{align}

Next we show that there exist  constants $\gamma >0$ and $n'_0$ (independent of $\kappa_1$, $\kappa_2$, and $\bd$) as well as some $\rho$ and $\lambda$ such that 
\begin{align}
 \min_{\bd \in \cB^n(v_n)} \min_{(\kappa_1,\kappa_2) \in \cW^{\el_n}_{\bd}} g^n_{\lambda, \rho}(\kappa_1,\kappa_2, \bd)  \geq \gamma, \quad n \geq n'_0. \label{Eq_detect_err_exp}
\end{align}
This then implies that $\Pr(\cE_{\kappa_1,\kappa_2}| \bD^a=\ \bd)$ vanishes as $n \to\infty$ uniformly in $(\kappa_1,\kappa_2) \in \cW^{\el_n}_{\bd}$ and $\bd \in \cB^n(v_n)$. Indeed, if $\bd \in \cB^n(v_n)$, then $|\bd| \leq v_n$ which implies that $\kappa_1 \leq v_n$. Furthermore, since the decoder outputs a vector in $\cB^n(v_n)$, we also have $\kappa_2 \leq v_n$.  It thus follows from~\eqref{Eq_err_mu_uppr}, \eqref{Eq_mu_lowr}, and~\eqref{Eq_detect_err_exp} that 
\begin{align}
\Pr(\cE_d| \bD^a=\ \bd) &\leq 2  v_n^2 \exp[-\tilde{E}_n \gamma ] \notag \\
& = 2 \exp\left[ -\tilde{E}_n \left(\gamma - \frac{ 2 \ln v_n}{\tilde{E}_n}\right)\right], \quad \bd \in \cB^n(v_n), n \geq \max(n_0,n'_0). \label{Eq_detect_err_uppr}
\end{align}
Furthermore, by the definition of $v_n$ and $\tilde{E}_n$,
\begin{align}
\frac{2\ln v_n}{\tilde{E}_n} & = \frac{2 \log e \log ((1+c_n)k_n)}{\tilde{E}_n} \notag \\
& =  \frac{4 \log e \log (1+c_n)}{b c_n \log \el_n} +  \frac{ 4 \log e \log k_n}{b c_n \log \el_n} \label{Eq_sn_En2}
\end{align}
which tends to zero since $c_n \to \infty$ and $1 \leq k_n \leq \el_n$.  Consequently, the RHS of~\eqref{Eq_detect_err_uppr} vanishes as $n \to \infty$.

To obtain~\eqref{Eq_detect_err_exp}, we first note that 
\begin{align}
\min_{\bd \in \cB^n(v_n)} \min_{(\kappa_1,\kappa_2) \in \cW^{\el_n}_{\bd}} g^n_{\lambda, \rho}(\kappa_1,\kappa_2, \bd)  & = \min \{ \min_{\bd \in \cB^n(v_n)}  \min_{1 \leq \kappa_1 \leq v_n} g^n_{\lambda, \rho}(\kappa_1,0,\bd), \min_{\bd \in \cB^n(v_n)}  \min_{ 1 \leq \kappa_2 \leq v_n } g^n_{\lambda, \rho}(0,\kappa_2, \bd), \notag \\
& \qquad \qquad \min_{\bd \in \cB^n(v_n)}  \min_{ \substack{ 1 \leq \kappa_1 \leq v_n \\ 1 \leq \kappa_2 \leq v_n}} g^n_{\lambda, \rho}(\kappa_1,\kappa_2, \bd)  \}. \label{Eq_inf_gn}
\end{align}
Then we show that for $\lambda = 2/3$ and $\rho = 3/4$,
\begin{align}
\liminf_{n \rightarrow \infty} \min_{\bd \in \cB^n(v_n)} \min_{1 \leq \kappa_1 \leq v_n} g^n_{\lambda, \rho}(\kappa_1, 0,\bd) & > 0 \label{Eq_w1_lowr} \\
\liminf_{n \rightarrow \infty} \min_{\bd \in \cB^n(v_n)} \min_{ 1 \leq \kappa_2 \leq v_n } g^n_{\lambda, \rho}(0,\kappa_2,\bd) & > 0 \label{Eq_w2_lowr} \\
\liminf_{n \rightarrow \infty} \min_{\bd \in \cB^n(v_n)} \min_{ \substack{ 1 \leq \kappa_1 \leq v_n \\ 1 \leq \kappa_2 \leq v_n}} g^n_{\lambda, \rho}(\kappa_1,\kappa_2, \bd)  & > 0  \label{Eq_w1w2_lowr}
\end{align}
from which~\eqref{Eq_detect_err_exp} follows.

Indeed, for $0 \leq \lambda \rho \leq 1$, we have
\begin{align}
& 2 \log \left(1+ \lambda(1-\lambda \rho)\kappa_2 \tilde{E}_n/n'' + \lambda \rho (1-\lambda \rho) \kappa_1 \tilde{E}_n/n''\right) \notag \\
& \qquad \qquad \geq  \log \left(1+ \lambda(1-\lambda \rho)\kappa_2 \tilde{E}_n/n'' \right) +  \log \left(1+\lambda \rho (1-\lambda \rho) \kappa_1 \tilde{E}_n/n''\right). \label{Eq_log_lowr}
\end{align}
Using~\eqref{Eq_log_lowr} in the second term on the RHS of~\eqref{Eq_err_exp}, we obtain that
\begin{align}
g^n_{\lambda, \rho}(\kappa_1, \kappa_2,\bd) & \geq  a^n_{\lambda, \rho}(\kappa_1,\bd)  + b^{n}_{\lambda, \rho}(\kappa_2)
\label{Eq_gn_lowr}
\end{align}
where 
\begin{align*}
 a^n_{\lambda, \rho}(\kappa_1,\bd)  \triangleq \frac{n''}{4\tilde{E}_n} \log \left(1 + \lambda \rho (1-\lambda \rho) \kappa_1 \tilde{E}_n/n''\right) - \frac{|\bd|}{\tilde{E}_n} H_2\left(\frac{\kappa_1}{|\bd|}\right)
\end{align*}
and 
\begin{align*}
b^{n}_{\lambda, \rho}(\kappa_2) \triangleq \frac{n''}{4\tilde{E}_n} \log \left(1+ \lambda(1-\lambda \rho)\kappa_2 \tilde{E}_n/n'' \right) -\frac{(1-\rho)}{2 \tilde{E}_n} \log (1+\lambda \kappa_2 \tilde{E}_n/n'') - \frac{\rho \el_n}{\tilde{E}_n} H_2\left(\frac{\kappa_2}{\el_n}\right).
\end{align*}

We begin by proving~\eqref{Eq_w1_lowr}. We have
\begin{align}
g^n_{\lambda, \rho}(\kappa_1, 0,\bd) & \geq a^n_{\lambda, \rho}(\kappa_1,\bd) + b^{n}_{\lambda, \rho}(0) \notag \\
& \geq  a^n_{\lambda, \rho}(\kappa_1,\bd) \label{Eq_gn_w1_lowr}
\end{align}
by~\eqref{Eq_gn_lowr} and $b^{n}_{\lambda, \rho}(0)  =0$.
Consequently,
\begin{align*}
 \min_{\bd \in \cB^n(v_n)}  \min_{1 \leq \kappa_1 \leq v_n} g^n_{\lambda, \rho}(\kappa_1, 0,\bd) & \geq  \min_{\bd \in \cB^n(v_n)}  \min_{1 \leq \kappa_1 \leq v_n} a^n_{\lambda, \rho}(\kappa_1,\bd)
\end{align*}
so~\eqref{Eq_w1_lowr} follows by showing that
\begin{align}
\liminf_{n \rightarrow \infty}  \min_{\bd \in \cB^n(v_n)}  \min_{1 \leq \kappa_1 \leq v_n} a^n_{\lambda, \rho}(\kappa_1,\bd) > 0. \label{Eq_first_lowr}
\end{align}
 To this end, let
\begin{align*}
i_n(\kappa_1) & \triangleq \frac{n''}{4\tilde{E}_n} \log \left(1 + \lambda \rho (1-\lambda \rho) \kappa_1 \tilde{E}_n/n''\right) \\
j_n(\kappa_1,\bd) & \triangleq \frac{|\bd|}{\tilde{E}_n} H_2\left(\frac{\kappa_1}{|\bd|}\right)
\end{align*}
so that
\begin{align}
 a^n_{\lambda, \rho}(\kappa_1,\bd) =  i_n(\kappa_1)   \left(1- \frac{j_n(\kappa_1,\bd)}{i_n(\kappa_1)}\right). \label{Eq_an}
\end{align}
Note that
\begin{align}
i_n(\kappa_1) & = \frac{n''}{4\tilde{E}_n} \log \left(1 + \lambda \rho (1-\lambda \rho) \kappa_1 \tilde{E}_n/n''\right) \notag \\
& \geq  \frac{n''}{4\tilde{E}_n} \log \left(1 + \lambda \rho (1-\lambda \rho) \tilde{E}_n/n''\right), \quad 1 \leq \kappa_1 \leq v_n \label{Eq_in_lowr}
\end{align}
and
\begin{align}
\frac{j_n(\kappa_1,\bd)}{i_n(\kappa_1)} & = \frac{4 |\bd| H_2\left(\frac{\kappa_1}{|\bd|}\right)}{n''  \log \left(1 + \lambda \rho (1-\lambda \rho) \kappa_1 \tilde{E}_n/n''\right) } \notag \\
& = \frac{4 \kappa_1 \log (|\bd|/\kappa_1) + |\bd|(\kappa_1/|\bd| - 1)   \log (1 -\kappa_1/|\bd|)   }{n''  \log \left(1 + \lambda \rho (1-\lambda \rho) \kappa_1 \tilde{E}_n/n''\right) }  \label{Eq_jn_in_ratio}.
\end{align}
Next we upper-bound $(\kappa_1/|\bd| - 1)   \log (1 -\kappa_1/|\bd|) $. Indeed, consider the function $f(p) = p - (p-1) \ln (1-p)$,
which satisfies $f(0)=0$ and is monotonically increasing in $p$.
 So, $(p-1) \ln (1-p) \leq p$, $0\leq p \leq 1$ which for $ p=\kappa_1/|\bd|$ gives 
\begin{align}
(\kappa_1/|\bd| - 1)   \log (1 -\kappa_1/|\bd|) \leq (\log e) \kappa_1/|\bd|. \label{Eq_fn_uppr}
\end{align}
Using~\eqref{Eq_fn_uppr} in~\eqref{Eq_jn_in_ratio}, we obtain that
\begin{align}
\frac{j_n(\kappa_1,\bd)}{i_n(\kappa_1)}  
& \leq \frac{ 4 \log (|\bd|/\kappa_1) +  4 \log e}{n''   \log \left(1 + \lambda \rho (1-\lambda \rho)  \kappa_1 \tilde{E}_n/n''\right)/\kappa_1 } \notag \\
& \leq \frac{4 \log (|\bd|/\kappa_1) +4 \log e }{n''   \log \left(1 + \lambda \rho (1-\lambda \rho) v_n \tilde{E}_n/n''\right)/v_n} \notag  \\
& \leq \frac{4 \log v_n +  4 \log e }{\tilde{E}_n \frac{ \log \left(1 + \lambda \rho (1-\lambda \rho) v_n \tilde{E}_n/n''\right)}{v_n\tilde{E}_n/n''}} \label{Eq_ratio_uppr4}  
\end{align}
where the second inequality follows because  $\frac{\log (1+x)}{x}$ is monotonically decreasing in $x > 0$, and the subsequent inequality follows because $|\bd| \leq v_n$ and $1 \leq \kappa_1 \leq v_n$. Combining~\eqref{Eq_an}, \eqref{Eq_in_lowr}, and~\eqref{Eq_ratio_uppr4}, $ a^n_{\lambda, \rho}(\kappa_1,\bd)$ can thus be lower-bounded by
\begin{align}
 a^n_{\lambda, \rho}(\kappa_1,\bd) &\geq   \frac{n''}{4\tilde{E}_n} \log \left(1 + \lambda \rho (1-\lambda \rho) \tilde{E}_n/n''\right)\left(1- \frac{4 \log v_n + 4 \log e }{\tilde{E}_n \frac{ \log \left(1 + \lambda \rho (1-\lambda \rho) v_n \tilde{E}_n/n''\right)}{v_n\tilde{E}_n/n''}}  \right). \label{Eq_an_lowr}
\end{align}
Note that the RHS of~\eqref{Eq_an_lowr} is independent of $\kappa_1$ and $\bd$. Furthermore, the term 
\begin{align}
\frac{ v_n \tilde{E}_n}{n''} & = \frac{ (1+c_n) k_n b c_n \ln \el_n}{2 b n} \notag \\
& =  \frac{  c_n k_n   \ln \el_n + c_n^2 k_n   \ln \el_n}{2  n}  \label{Eq_En_sn_zero2} 
\end{align}
 tends to zero as $n \to \infty$ since $k_n \ln \el_n = o(n)$ and $c_n = \ln\left(\frac{n}{k_n\ln \el_n}\right)$. 
Furthermore, $\tilde{E}_n\to\infty$ and, as observed in~\eqref{Eq_sn_En2}, $\frac{\log v_n}{\tilde{E}_n} \to 0$ as $n \to \infty$. It follows that
\begin{align}
& \liminf_{n \to \infty}   \min_{\bd \in \cB^n(v_n)}   \min_{1 \leq \kappa_1 \leq v_n} a^n_{\lambda, \rho}(\kappa_1,\bd) \notag \\
& \qquad \qquad \geq \lim_{n \to \infty} \frac{n''}{4\tilde{E}_n} \log \left(1 + \lambda \rho (1-\lambda \rho) \tilde{E}_n/n''\right)   \lim_{n \to \infty} \left(1- \frac{ \log v_n +  \log e }{\tilde{E}_n \frac{ \log \left(1 + \lambda \rho (1-\lambda \rho) v_n \tilde{E}_n/n''\right)}{v_n\tilde{E}_n/n''}}  \right) \notag \\
& \qquad \qquad= \frac{(\log e) \;  \lambda \rho (1-\lambda \rho)}{4} \notag \\
&\qquad \qquad = \frac{\log e}{16} \notag
\end{align}
which implies~\eqref{Eq_w1_lowr}.

We next prove~\eqref{Eq_w2_lowr}. Since $a^{n}_{\lambda, \rho}(0,\bd)  =0 $, it follows that
\begin{align*}
 \min_{\bd \in \cB^n(v_n)}  \min_{1 \leq \kappa_2 \leq v_n} g^n_{\lambda, \rho}(0, \kappa_2,\bd) & \geq \min_{1 \leq \kappa_2 \leq v_n} b^n_{\lambda, \rho}(\kappa_2).
\end{align*}
Thus~\eqref{Eq_w2_lowr} follows by showing that
\begin{align}
\liminf_{n \rightarrow \infty} \min_{1 \leq \kappa_2 \leq v_n} b^n_{\lambda, \rho}(\kappa_2) > 0. \label{Eq_w2_lowr1}
\end{align}
To show~\eqref{Eq_w2_lowr1}, we define
\begin{align}
q_n(\kappa_2) & \triangleq \frac{n''}{4\tilde{E}_n} \log \left(1+ \lambda(1-\lambda \rho)\kappa_2 \tilde{E}_n/n'' \right) \label{Eq_qn} \\
r_n(\kappa_2) & \triangleq \frac{(1-\rho)}{2 \tilde{E}_n} \log (1+\lambda \kappa_2 \tilde{E}_n/n'')  \\
u_n(\kappa_2) & \triangleq \frac{\rho \el_n}{\tilde{E}_n} H_2\left(\frac{\kappa_2}{\el_n}\right). \label{Eq_sn}
\end{align}
Then, 
\begin{align}
b^n_{\lambda, \rho}(\kappa_1) & =  q_n(\kappa_2)  \left(1- \frac{r_n(\kappa_2)}{q_n(\kappa_2)} - \frac{u_n(\kappa_2)}{q_n(\kappa_2)}\right). \label{Eq_bn_lowr}
\end{align}
Note that
\begin{align}
q_n(\kappa_2) & =   \frac{n''}{4\tilde{E}_n} \log \left(1+ \lambda(1-\lambda \rho)\kappa_2 \tilde{E}_n/n'' \right) \notag \\
& \geq  \frac{n''}{4\tilde{E}_n} \log \left(1+ \lambda(1-\lambda \rho) \tilde{E}_n/n'' \right), \quad 1 \leq \kappa_2 \leq v_n. \label{Eq_qn_lowr}
\end{align}
Furthermore,
\begin{align}
\frac{r_n(\kappa_2)}{ q_n(\kappa_2)} &= \frac{\frac{(1-\rho)}{2 \tilde{E}_n} \log (1+\lambda \kappa_2 \tilde{E}_n/n'')}{\frac{n''}{4\tilde{E}_n} \log \left(1+ \lambda(1-\lambda \rho)\kappa_2 \tilde{E}_n/n'' \right)} \notag \\
& \leq  \frac{\frac{(1-\rho)}{2 \tilde{E}_n} \log (1+\lambda v_n \tilde{E}_n/n'')}{\frac{n''}{4\tilde{E}_n} \log \left(1+ \lambda(1-\lambda \rho) \tilde{E}_n/n'' \right)} \notag \\
& \leq  \frac{\frac{(1-\rho)v_n}{2 n'' } \frac{ \log (1+\lambda v_n \tilde{E}_n/n'')}{\tilde{E}_nv_n/n''}}{ \frac{ \log \left(1+ \lambda(1-\lambda \rho) \tilde{E}_n/n'' \right)}{4\tilde{E}_n/n''}} \label{Eq_rn_qn_ratio}
\end{align}
for $1 \leq \kappa_2 \leq v_n$.
The term 
\begin{align*}
\frac{v_n}{n''} & = \frac{ k_n(1+c_n)}{bn} 
\end{align*}
tends to zero since $k_n c_n = o(n)$ by the lemma's assumption that $k_n \log \el_n = o(n)$. This together with the fact that $v_n \tilde{E}_n/n'' \to 0$ as $n \to \infty$ (cf.~\eqref{Eq_En_sn_zero2}),  and hence  $\tilde{E}_n/n'' \to 0$, implies that the RHS of~\eqref{Eq_rn_qn_ratio} tends to zero as $n \to \infty$.
Finally,
\begin{align}
\frac{u_n(\kappa_2)}{ q_n(\kappa_2)}  & = \frac{4 \rho \el_n H_2\left(\frac{\kappa_2}{\el_n}\right)}{n'' \log \left(1+ \lambda(1-\lambda \rho)\kappa_2 \tilde{E}_n/n'' \right)} \notag \\
& = \frac{4 \rho \left[ \kappa_2 \log (\el_n/\kappa_2) + \el_n (\kappa_2/\el_n -1)   \log (1 -\kappa_2/\el_n) \right]  }{n'' \log \left(1+ \lambda(1-\lambda \rho)\kappa_2 \tilde{E}_n/n'' \right)}\notag\\
& \leq \frac{4 \rho \left[ \kappa_2 \log (\el_n/\kappa_2) + \kappa_2 \log e \right]  }{n'' \log \left(1+ \lambda(1-\lambda \rho)\kappa_2 \tilde{E}_n/n'' \right)} \notag\\
& \leq  \frac{4 \rho \left[  \log (\el_n/\kappa_2) +  \log e \right]  }{n'' \log \left(1+ \lambda(1-\lambda \rho)v_n \tilde{E}_n/n'' \right)/v_n} \notag \\
& \leq  \frac{4 \rho \left[  \log \el_n +  \log e \right]  }{\tilde{E}_n \frac{ \log \left(1+ \lambda(1-\lambda \rho)v_n \tilde{E}_n/n'' \right)}{v_n\tilde{E}_n/n''}} \label{Eq_sn_qn_uppr}
\end{align}
for $1 \leq \kappa_2 \leq v_n$, where the first inequality follows from~\eqref{Eq_fn_uppr}. 
Since  $\frac{\log \el_n}{ \tilde{E}_n} = \frac{\log \el_n}{c_n \ln \el_n} \to 0$ and $v_n\tilde{E}_n/n'' \to 0$ as $n \to \infty$, the RHS of~\eqref{Eq_sn_qn_uppr} tends to zero as $n \to \infty$. Thus, it follows from~\eqref{Eq_qn_lowr}, \eqref{Eq_rn_qn_ratio}, and \eqref{Eq_sn_qn_uppr} that
\begin{align}
 b^n_{\lambda, \rho}(\kappa_2) & \geq \frac{n''}{4\tilde{E}_n} \log \left(1+ \lambda(1-\lambda \rho) \tilde{E}_n/n'' \right)\left(1 - \frac{\frac{(1-\rho)v_n}{2 n'' } \frac{ \log (1+\lambda v_n \tilde{E}_n/n'')}{\tilde{E}_nv_n/n''}}{ \frac{ \log \left(1+ \lambda(1-\lambda \rho) \tilde{E}_n/n'' \right)}{4\tilde{E}_n/n''}}  - \frac{4 \rho \left[  \log \el_n +  \log e \right]  }{\tilde{E}_n \frac{ \log \left(1+ \lambda(1-\lambda \rho)v_n \tilde{E}_n/n'' \right)}{v_n\tilde{E}_n/n''}} \right). \label{Eq_bn_lowr_bnd}
\end{align}
The lower bound in~\eqref{Eq_bn_lowr_bnd} is independent of $\kappa_2$ and $\bd$. It thus follows that
\begin{align}
&\liminf_{n \rightarrow \infty}   \min_{1 \leq \kappa_2 \leq v_n} b^n_{\lambda, \rho}(\kappa_2) \notag \\
&  \geq \lim_{n \to \infty}   \frac{n''}{4\tilde{E}_n} \log \left(1+ \lambda(1-\lambda \rho) \tilde{E}_n/n'' \right) \lim_{n \to \infty} \! \left(1 - \frac{\frac{(1-\rho)v_n}{2 n'' } \frac{ \log (1+\lambda v_n \tilde{E}_n/n'')}{\tilde{E}_nv_n/n''}}{ \frac{ \log \left(1+ \lambda(1-\lambda \rho) \tilde{E}_n/n'' \right)}{4\tilde{E}_n/n''}}  - \frac{4 \rho \left[  \log \el_n +  \log e \right]  }{\tilde{E}_n \frac{ \log \left(1+ \lambda(1-\lambda \rho)v_n \tilde{E}_n/n'' \right)}{v_n\tilde{E}_n/n''}} \right) \! \label{Eq_lim_bn} \notag \\
&  =  \frac{(\log e)  \; \lambda(1-\lambda \rho)}{4} \notag \\
&  = \frac{\log e}{12}\notag
\end{align}
which implies~\eqref{Eq_w2_lowr}.

To prove~\eqref{Eq_w1w2_lowr},  we use~\eqref{Eq_gn_lowr}, \eqref{Eq_an_lowr}, and \eqref{Eq_bn_lowr_bnd} to lower-bound
\begin{align}
 g^n_{\lambda, \rho}(\kappa_1,\kappa_2, \bd)  & \geq  \frac{n''}{4\tilde{E}_n} \log \left(1 + \lambda \rho (1-\lambda \rho) \tilde{E}_n/n''\right)\left(1- \frac{ \log v_n +  \log e }{\tilde{E}_n \frac{ \log \left(1 + \lambda \rho (1-\lambda \rho) v_n \tilde{E}_n/n''\right)}{v_n\tilde{E}_n/n''}}  \right) \notag \\
 & + \frac{n''}{4\tilde{E}_n} \log \left(1+ \lambda(1-\lambda \rho) \tilde{E}_n/n'' \right)\left(1 - \frac{\frac{(1-\rho)v_n}{2 n'' } \frac{ \log (1+\lambda v_n \tilde{E}_n/n'')}{\tilde{E}_nv_n/n''}}{ \frac{ \log \left(1+ \lambda(1-\lambda \rho) \tilde{E}_n/n'' \right)}{4\tilde{E}_n/n''}}  - \frac{4 \rho \left[  \log \el_n +  \log e \right]  }{\tilde{E}_n \frac{ \log \left(1+ \lambda(1-\lambda \rho)v_n \tilde{E}_n/n'' \right)}{v_n\tilde{E}_n/n''}} \right) \notag
\end{align}
which is independent of $\kappa_1, \kappa_2$, and $\bd$. Consequently,
\begin{align}
& \liminf_{n \to \infty}  \min_{\bd \in \cB^n(v_n)}  \min_{ \substack{ 1 \leq \kappa_1 \leq v_n \\ 1 \leq \kappa_2 \leq v_n}} g^n_{\lambda, \rho}(\kappa_1,\kappa_2, \bd)  \notag \\
 & \geq \lim_{n \to \infty} \left\{ \frac{n''}{4\tilde{E}_n} \log \left(1 + \lambda \rho (1-\lambda \rho) \tilde{E}_n/n''\right)\left(1- \frac{ \log v_n +  \log e }{\tilde{E}_n \frac{ \log \left(1 + \lambda \rho (1-\lambda \rho) v_n \tilde{E}_n/n''\right)}{v_n\tilde{E}_n/n''}}  \right)\right\} \notag \\
& + \lim_{n \to \infty}  \left\{ \frac{n''}{4\tilde{E}_n} \log \left(1+ \lambda(1-\lambda \rho) \tilde{E}_n/n'' \right)\left(1 - \frac{\frac{(1-\rho)v_n}{2 n'' } \frac{ \log (1+\lambda v_n \tilde{E}_n/n'')}{\tilde{E}_nv_n/n''}}{ \frac{ \log \left(1+ \lambda(1-\lambda \rho) \tilde{E}_n/n'' \right)}{4\tilde{E}_n/n''}}  - \frac{4 \rho \left[  \log \el_n +  \log e \right]  }{\tilde{E}_n \frac{ \log \left(1+ \lambda(1-\lambda \rho)v_n \tilde{E}_n/n'' \right)}{v_n\tilde{E}_n/n''}} \right)\right\} \notag \\
& = \frac{\log e}{12} + \frac{\log e}{16} \notag
\end{align}
which implies~\eqref{Eq_w1w2_lowr}. This was the last step required to prove~\eqref{Eq_detect_err_exp}.

We finish the proof of Lemma 1 by analyzing the third term on the RHS of~\eqref{Eq_detect_err_uoor}, namely,   $\Pr(\cE_d| |\bD^a| = 0) \Pr(|\bD^a| = 0)$. This term is upper-bounded by
\begin{align*}
 \Pr(\cE_d| |\bD^a| = 0)\left((1-\alpha_n)^{\frac{1}{\alpha_n}}\right)^{k_n}
\end{align*}
 and vanishes if $k_n$ is unbounded. Next we show that this term also vanishes when  $k_n$ is bounded. When $|\bD^a| = 0$, an error occurs only if there are false alarms. For $\kappa_2$ false alarms, let  $\bar{\bS} \triangleq \sum_{j=1}^{\kappa_2} \bS_j$, and let $S'_i$ denote the $i$-th component of $\bar{\bS}$. From~\cite[eq.~(303)]{ChenCG17},  we obtain  the following upper bound on the probability that there are $\kappa_2$ false alarms when $|\bD^a| = 0$:
\begin{align*}
P(\cE_{\kappa_2}|  |\bd| =0) & \leq  \binom{\el_n}{\kappa_2} \mathrm{E}_{\underline{\bS}_{A_2}} \left\{ \Pr\left\{ \sum_{i=1}^{n''} Z_i S'_i\geq \frac{1}{2} \|\bar{\bS} \|^2 \right\} |\bar{\bS} \right\} \notag  \\
& \leq \left(\frac{1}{\mu}\right)^{\kappa_2}  \binom{\el_n}{\kappa_2} \mathrm{E}_{\tilde{\underline{\bS}}_{A_2}} \left\{ \Pr\left\{ \sum_{i=1}^{n''} Z_i S'_i \geq \frac{1}{2} \|\bar{\bS} \|^2 \right\} |\bar{\bS} \right\} 
\end{align*}
where in the last inequality, we used~\eqref{Eq_prob_signt_uppr}. By following the analysis that led to~\cite[eq.~(309)]{ChenCG17}, we obtain
\begin{align}
P(\cE_{\kappa_2}|  |\bd| =0)  & \leq  \left(\frac{1}{\mu}\right)^{\kappa_2}   \exp \left[ -\tilde{E}_n (q'_n(\kappa_2) - u_n'(\kappa_2) ) \right] \notag
\end{align}
where 
\begin{align}
q'_n(\kappa_2) &\triangleq  \frac{n''}{2\tilde{E}_n}\log \left( 1+ \frac{\kappa_2 \tilde{E}_n}{4 n''} \right) \notag
\end{align}
and 
\begin{align}
u'_n(\kappa_2) & \triangleq  \frac{ \el_n}{\tilde{E}_n} H_2\left(\frac{\kappa_2}{\el_n}\right). \notag
\end{align}
As before, we upper-bound $ \left(\frac{1}{\mu}\right)^{\kappa_2} \leq 2$ uniformly in $\kappa_2$  for $n \geq \tilde{n}_0$. Furthermore, we observe that the behaviours of $q'_n(\kappa_2)$ and $u'_n(\kappa_2)$ are similar to $q_n(\kappa_2)$ and $v_n(\kappa_2)$ given in~\eqref{Eq_qn} and in~\eqref{Eq_sn}, respectively. So by following the steps as before, we can show that
\begin{align}
\liminf_{n \to \infty} \min_{1 \leq \kappa_2 \leq v_n} q'_n(\kappa_2) > 0  \notag
\end{align}
and 
\begin{align}
\lim_{n \to \infty}   \min_{1 \leq \kappa_2 \leq v_n} \frac{u'_n(\kappa_2)}{ q'_n(\kappa_2)} = 0. \notag
\end{align}
It follows that there exist positive constants $\tau'$ and $ \tilde{n}'_0$ (independent of $\kappa_2$) such that,
\begin{align}
P(\cE_d |  |\bd| =0)  & = \sum_{\kappa_2=1}^{v_n}P(\cE_{\kappa_2} |  |\bd| =0)\\
& \leq 2 v_n \exp \left[ -\tilde{E}_n \tau' \right], \quad n\geq \max(\tilde{n}_0, \tilde{n}'_0). \notag
\end{align}
We have already shown that $ v_n^2 \exp [ -\tilde{E}_n \tau' ]$ vanishes as $n \to \infty$ (cf. \eqref{Eq_detect_err_uppr}--\eqref{Eq_sn_En2}), which implies that $ 2 v_n \exp [ -\tilde{E}_n \tau' ]$  vanishes as $n \to \infty$. It thus follows that $P(\cE_d |  |\bd| =0)$ tends to zero as $n \to \infty$. This was the last step required to prove Lemma~\ref{Lem_usr_detect}.

\section{Proof of Lemma~\ref{Lem_energy_bound}}
\label{Append_prob_lemma}

Let $\cW$ denote the set of  $(M_n+1)^{\el_n}$ messages  of all users. To prove Lemma~\ref{Lem_energy_bound}, we represent each $\bw \in \cW$   using an $\el_n$-length vector such that the $i^{\mathrm{th}}$ position of the vector is set to $j$ if user $i$ has message $j$. 
The Hamming distance $d_H$ between two messages $\bw=(w_1,\ldots,w_{\el_n})$ and $\bw'=(w'_1,\ldots,w'_{\el_n})$ is defined as the number of positions at which $\bw$ differs from $\bw'$, i.e.,
$d_H(\bw,\bw') := \left|\{i: w_i\neq w'_i \}\right|$.

We first group the set $\cW$ into $\el_n +1$ subgroups. 
Two messages $\bw, \bw' \in \cW$ belong to the same subgroup if they have the same number of zeros. We can observe that all the messages in a subgroup have the same probability since the probability of a message $\bw$ is determined by the number of zeros in it.

Let $\cT_{\bw}^{t} $ denote the set of all messages with $t$ non-zero entries, where $t=0, \ldots, \el_n$. 
Further let 
\begin{align}
\text{Pr}(\cT_{\bw}^{t} ) \triangleq  \text{Pr}(\bW \in \cT_{\bw}^{t} )\notag
\end{align}
which can be evaluted as
\begin{equation}
\text{Pr}(\cT_{\bw}^{t} ) = (1-\alpha_n)^{\el_n-t} \left( \frac{\alpha_n}{M_n}\right)^{t} |\cT_{\bw}^{t} |. \label{Eq_type_prob}
\end{equation}
We define
\begin{align}
P_e(\cT_{\bw}^{t} ) \triangleq \frac{1}{|\cT_{\bw}^{t} |} \sum_{w\in \cT_{\bw}^{t} } P_e(\bw) \label{Eq_type_err_prob}
\end{align}
where $P_e(\bw)$ denotes the probability of error in decoding the set of messages $\bw=(w_1,\ldots,w_{\el_n})$.
It follows that
\begin{align}
P_{e}^{(n)} & =  \sum_{w \in \cW}  \text{Pr}(\bW =\bw) P_e(\bw) \notag \\
& =  \sum_{t=0}^{\el_n} \sum_{w\in \cT_{\bw}^{t} } (1-\alpha_n)^{\el_n-t} \left( \frac{\alpha_n}{M_n}\right)^{t} |\cT_{\bw}^{t} | \frac{1}{|\cT_{\bw}^{t} |}  P_e(\bw)\notag \\
& =  \sum_{t=0}^{\el_n}  \text{Pr}( \cT_{\bw}^{t} ) \frac{1}{|\cT_{\bw}^{t} |} \sum_{w\in \cT_{\bw}^{t} } P_e(\bw) \notag \\
& =  \sum_{t=0}^{\el_n}  \text{Pr}(\cT_{\bw}^{t} ) P_e(\cT_{\bw}^{t} ) \notag \\
& \geq \sum_{t=1}^{\el_n}  \text{Pr}(\cT_{\bw}^{t} ) P_e(\cT_{\bw}^{t} ) \label{Eq_avg_prob_err3}
\end{align}
where we have used \eqref{Eq_type_prob} and the definition of $P_e(\mathcal{T}_{\mathbf{w}}^t)$ in \eqref{Eq_type_err_prob}.
To prove Lemma~\ref{Lem_energy_bound}, we next show that 
\begin{align}
P_e(\cT_{\bw}^{t} ) \geq  1  -   \frac{ 256 E_n/N_0+\log 2}{\log \el_n}, \quad t=1,\ldots, \el_n. \label{Eq_prob_typ_lowr}
\end{align}
To this end, we partition each $\cT_{\bw}^{t},t=1,\ldots, \el_n  $ into $D_t$ sets $\cS_d^t$.  For every $ t=1,\ldots, \el_n $, the partition that we obtain satisfies 
\begin{align}
\frac{1}{|\cS_d^t|} \sum_{\bw \in \cS_d^t} P_e(\bw) \geq 1  -   \frac{ 256 E_n/N_0+\log 2}{\log \el_n}. \label{Eq_avg_prob_lowr}
\end{align}
This then gives~\eqref{Eq_prob_typ_lowr} since
\begin{align}
P_e(\cT_{\bw}^{t} ) & = \sum_{d=1}^{D_t} \frac{|\cS_d^t|}{|\cT_{\bw}^{t} |} \frac{1}{|\cS_d^t|} \sum_{w \in \cS_d^t} P_e(\bw).
\end{align}

Before we continue by defining the sets $\mathcal{S}_d^t$, we note that
\begin{align}
M_n \geq  2 \label{Eq_seqMn_assum}
\end{align}
since $M_n=1$ would contradict the assumption that $\CR>0$.
We further assume that
\begin{align}
\el_n \geq 5. \label{Eq_seqln_assum}
\end{align}
This assumption comes without loss of generality 
since $\el_n \to \infty$ as $n \to \infty$ by the assumption  that $k_n \log \el_n = \omega(n)$ and $k_n = \Omega(1)$. 

We next define a partition of $\cT_{\bw}^{t}, t=1,\ldots, \el_n  $  that satisfies the following:
\begin{align}
|\cS_d^t| \geq \el_n +1, \quad  d=1, \ldots, D_t \label{Eq_set_lowr}
\end{align}
and
\begin{align}
d_{H}(\bw,\bw') \leq 8, \quad \bw, \bw \in  \cS_d^t. \label{Eq_distnce_uppr}
\end{align}
To this end, we consider the following four cases:

\underline{ Case 1: $t=1$:} For $t=1$, we do not partition the set, i.e., $\cS_1^1 = \cT_ {\bw}^{1}$. Thus, we have $|\cS_1^1| = \el_nM_n$. From~\eqref{Eq_seqMn_assum} and~\eqref{Eq_seqln_assum}, it follows that $|\cS_1^1| \geq \el_n +1$. Since any two messages $\bw, \bw' \in  \cT_ {\bw}^{1}$ 
have only one non-zero entry, we further have that $d_{H}(\bw,\bw') \leq 2$. Consequently, \eqref{Eq_set_lowr} and~\eqref{Eq_distnce_uppr} are satisfied.

\underline{ Case 2: $t=2,\ldots,\el_n -2$:}
In this case, we obtain a partition by finding a code $\cC_t$ in $\cT_ {\bw}^{t}$  that has minimum Hamming distance $5$ and for every $\bw\in \cT_ {\bw}^{t}$ there exists at least one codeword in $\cC_t$ which is at most at a Hamming distance 4 from it. 
Such a code exists because if for some $\bw\in \cT_ {\bw}^{t}$ all codewords were at a Hamming distance 5 or more, then we could add $\bw$ to $\cC_t$ without affecting its minimum distance. 
Thus for all $\bw \notin \cC_t$, there exists at least one index $j$ such that $d_H(\bw,\bc_t(j)) \leq 4$, where
$\bc_t(1),\ldots, \bc_t(|\cC_t|)$ denote the codewords of code $\cC_t$. With this code $\mathcal{C}_t$, we partition $\mathcal{T}_{\mathbf{w}}^t$ into the sets $\mathcal{S}_d^t$, $d=1,\ldots,D_t$ using the following procedure:

For a given $d=1, \ldots,D_t$, we assign $\bc_t(d)$ to $\cS_d^t$ as well as all $\bw \in \cT_ {\bw}^{t}$ that satisfy $d_H(\bw, \bc_t(d))\leq 2$. These assignments are unique since the code $\cC_t$ has minimum Hamming distance 5. We next consider  all $\bw\in  \cT_ {\bw}^{t}$ for which there is no codeword $\bc_t(1), \ldots, \bc_t(|\cC_t|)$ satisfying $d_H(\bw, \bc_t(d))\leq  2$ and assign it  to the set with index $d = \min \{j=1,\ldots, D_t: d_H(\bw, \bc_t(j)) \leq 4 \}$. Like this, we obtain a partition of $ \cT_ {\bw}^{t}$.  
Since any two $\bw, \bw' \in  \cS_d^t$ are at most at a Hamming distance 4 from the codeword $\bc_t(d)$, we have that
$d_{H}(\bw,\bw') \leq 8$. Consequently, \eqref{Eq_distnce_uppr} is satisfied.

To show that~\eqref{Eq_set_lowr} is satisfied, too, we use the following fact:
\begin{align}
\text{For two natural numbers } a \text{ and } b, \text { if } a \geq 4 \text{ and } 2\leq b \leq a-2, \text{ then } b(a-b) \geq a. \label{Eq_prod_seq}
\end{align}
This fact follows since $b(a-b)$ is increasing from $b=1$ to $b = \lfloor a/2\rfloor$ and is 
decreasing from $b = \lfloor a/2\rfloor$ to $b=a-1$. So $b(a-b)$ is minimized at $b=2$ and $b=a-2$, where it has the value $2a-4$. For $a\geq 4$, this value is greater than or equal to $a$, hence the claim follows.

From~\eqref{Eq_prod_seq}, it follows that if   $|\cS_d^t| \geq 1+ t(\el_n - t)$, then $|\cS_d^t|\geq 1+ \el_n$. It thus remains to show that  $|\cS_d^t| \geq 1+  t(\el_n - t)$.
To this end, for every codeword $\bc_t(d)$, consider all sequences in $\cT_ {\bw}^{i}$ which differ exactly in one non-zero position and in one zero position from $\bc_t(d)$. There are $ t(\el_n - t)M_n$ such  sequences  in $\cT_ {\bw}^{t}$, so we get
\begin{align}
t(\el_n - t)M_n & \geq  t(\el_n - t) \notag \\
& \geq \el_n  \label{Eq_set_lowr2}
\end{align}
by~\eqref{Eq_seqMn_assum}, \eqref{Eq_seqln_assum}, and~\eqref{Eq_prod_seq}.
Since the codeword $\bc_t(d)$ also belongs to $S_d^t$, it follows from~\eqref{Eq_set_lowr2} that 
\begin{align}
| \cS_d^t| &\geq \el_n +1. \notag
\end{align}

\underline{ Case 3: $t=\el_n -1$:}
 We obtain a partition by defining a code $\cC_t$ in $ \cT_ {\bw}^{\el_n -1}$ 
that has 
the same properties as the code used for Case 2. We then use the same procedure as in Case 2 to assign messages in $\bw \in\cT_ {\bw}^{\el_n -1}$ to the sets $\cS_d^t$, $d=1,\ldots,D_t$. This gives a partition of $\cT_ {\bw}^{\el_n -1}$ where any two $\bw, \bw'\in\cS_d^t$ satisfy $d_H(\mathbf{w},\mathbf{w}')\leq 8$. Consequently, this partition satisfies~\eqref{Eq_distnce_uppr}. 

We next show that this partition also satisfies~\eqref{Eq_set_lowr}. To this end, for every codeword $\bc_t(d)$, consider all the sequences which differ exactly in two non-zero positions from $\bc_t(d)$.
There are $ \binom{\el_n-1}{2} (M_n-1)^2$  such sequences in $ \cT_ {\bw}^{\el_n -1}$. Since $\cS_d^t$ also contains the codeword $\bc_t(d)$, we obtain that
\begin{align*}
| \cS_d^t|  & \geq \notag
\binom{\el_n-1}{2} (M_n-1)^2 + 1\\
& \geq \binom{\el_n-1}{2} +1   \\
& \geq  \el_n +1 
\end{align*}  
 by~\eqref{Eq_seqMn_assum} and~\eqref{Eq_seqln_assum}.

\underline{ Case 4: $t=\el_n$:} We obtain a partition by defining a code $\mathcal{C}_t$ in $\mathcal{T}_{\mathbf{w}}^{\el_n-1}$ that has the same properties as the code used in Case 2. We then use the same procedure as in Case 2 to assign messages in $\mathbf{w}\in\mathcal{T}_{\mathbf{w}}^t$ to the sets $\mathcal{S}_d^t$, $d=1,\ldots,D_t$. This gives a partition of $\mathcal{T}_{\mathbf{w}}^t$ where any two $\mathbf{w}, \mathbf{w}'\in\mathcal{S}_d^t$ satisfy $d_H(\mathbf{w},\mathbf{w}')\leq 8$. Consequently, this partition satisfies~\eqref{Eq_distnce_uppr}.

We next show that this partition also satisfies~\eqref{Eq_set_lowr}. To this end, for every codeword $\bc_t(d)$, consider all sequences which are  at Hamming distance $1$ from $\bc_t(d)$. There are $\el_n(M_n-1)$ such sequences. Since $\cS_d^t$ also contains the codeword, we have 
\begin{align}
| \cS_d^t| & \geq 1+ \el_n(M_n-1) \notag \\
& \geq 1+\el_n \notag
\end{align}  
by~\eqref{Eq_seqMn_assum}. 

Having obtained a partition of $\mathcal{T}_{\mathbf{w}}^t$ that satisfies~\eqref{Eq_set_lowr} and~\eqref{Eq_distnce_uppr}, we next derive the lower bound~\eqref{Eq_avg_prob_lowr}.	To this end, we use  a stronger form of Fano's inequality known as Birg\'e's inequality.
\begin{lemma}[Birg\'e's inequality]
	\label{Lem_Berge}
	Let $(\cY, \cB)$ be a measurable space with a $\sigma$-field, and let $P_1,\ldots, P_N$  be probability measures defined on $\cB$. Further let $\cA_i$, $i=1, \ldots,N$  denote $N$ events defined on $\cY$, where $N\geq 2$.
	Then 
	\begin{align*}
	\frac{1}{N} \sum_{i=1}^{N} P_i(\cA_i) \leq \frac{\frac{1}{N^2} \sum_{i,j}^{}  D(P_i\|P_j)+\log 2}{\log (N-1)}.
	\end{align*}
\end{lemma}
\begin{proof}
	See~\cite{Yatracos88} and references therein.
\end{proof}

To apply Lemma~\ref{Lem_Berge} to the problem at hand, we set $N=|\cS_d^t|$ and $P_j = P_{Y|\bX}(\cdot|\bx(j))$, where $\bx(j)$ denotes the set of codewords transmitted to convey the set of messages $j \in \cS_d^t$.
We further define $\cA_j$ as the subset of $\cY^n$ for which the decoder declares the set of messages $j\in\cS_d^t$. Then, the probability of
error in decoding messages $j\in\cS_d^t$ is given by $P_e(j) =1-P_j(\cA_j)$, and $\frac{1}{|\cS_d^t|} \sum_{j\in \cS_d^t} P_j(\cA_j)$ denotes the average probability of correctly decoding a message in $\cS_d^t$.

For two multivariate Gaussian distributions \mbox{${\bf Z}_1 \sim \cN(\boldsymbol {\mu_1 }, \frac{N_0}{2}I)$}
and ${\bf Z}_2 \sim \cN(\boldsymbol {\mu_2}, \frac{N_0}{2}I)$ (where $I$ denotes the identity matrix),
the relative entropy $D({\bf Z}_1\| { \bf Z}_2)$ is given  by $ \frac{ ||\boldsymbol {\mu_1 - \mu_2}||^2}{N_0}$. We next note that $P_{\bw} =  \cN(\overline{\bx}(\bw), \frac{N_0}{2}I)$ and $P_{\bw'} = \cN(\overline{\bx}(\bw'), \frac{N_0}{2}I)$, where $\overline{\bx}(j)$ denotes the sum of codewords contained in $\bx(j)$.
Furthermore, any two messages $\bw, \bw' \in \cS_d^t$ are at a Hamming distance of at most 8. Without loss of generality, let us assume that $w_j = w'_j$ for $j=9, \ldots, \el_n$. Then
\begin{align}
\left\|\sum_{j=1}^{\el_n} \bx_j(w_j) -\sum_{i=1}^{\el_n} \bx_j(w'_j)\right\|^2 & = \left\|\sum_{i=1}^{8} \bx_j(w_j) - \bx_j(w'_j)\right\|^2 \notag \\
& \leq \left\|\sum_{j=1}^{8} |\bx_j(w_j) - \bx_j(w'_j)|\right\|^2 \notag\\
& \leq (8 \times 2\sqrt{E_n})^2 \notag \\
& = 256 E_n \notag
\end{align}
where we have used the triangle inequality and that the energy of a codeword for any user is upper-bounded by $E_n$. Thus, $D(P_{\bw}\| P_{\bw'}) \leq 256 E_n/N_0$.

It  follows from Birg\'e's inequality that
\begin{align}
\frac{1}{|\cS_d^t|} \sum_{\bw \in \cS_d^t} P_e(\bw) & \geq 1  -   \frac{ 256 E_n/N_0+\log 2}{\log (|\cS_d^t|-1)} \notag \\
&  \geq 1  -   \frac{ 256 E_n/N_0+\log 2}{\log \el_n } \label{Eq_prob_set_lowr}
\end{align}
where the last step holds because  $ |\cS_d^t|-1 \geq \el_n$. This proves~\eqref{Eq_avg_prob_lowr} and hence also~\eqref{Eq_prob_typ_lowr}.

Combining~\eqref{Eq_prob_typ_lowr} and~\eqref{Eq_avg_prob_err3}, we obtain
\begin{align*}
P_{e}^{(n)} & \geq \left( 1  -   \frac{ 256 E_n/N_0+\log 2}{\log \el_n } \right)  \sum_{i=1}^{\el_n}  \text{Pr}(\cT_ {\bw}^{i}) \\
&  = \left( 1  -   \frac{ 256 E_n/N_0+\log 2}{\log \el_n } \right) (1-\text{Pr}(\cT_{\bw}^{0} )). 
\end{align*}
By the assumption $k_n = \Omega(1)$, the probability 
$\text{Pr}(\cT_{\bw}^{0} ) = \left( (1-\alpha_n)^{\frac{1}{\alpha_n}}\right)^{k_n}$ converges to a value strictly less than one.  Consequently, $P_{e}^{(n)}$ tends to zero only if 
\begin{align}
E_n  & = \Omega\left(\log \el_n \right).\notag
\end{align}
This proves Lemma~\ref{Lem_energy_bound}.

\section{Proof of Lemma~\ref{Lem_detect_uppr}}
\label{Sec_appnd_ortho}
Let $\bY_1$ denote the received vector of length $n/\el_n$ corresponding to user 1 in the orthogonal-access scheme.
 From the pilot signal, which is the first symbol $Y_{11} $ of $\bY_1$, the receiver guesses whether user 1 is active or not. Specifically, the user is estimated as active if $Y_{11} > \frac{\sqrt{tE_n}}{2}$ and as inactive otherwise.
If the user is declared as active, then the receiver decodes the message from the rest of $\bY_1$.
Let $\Pr( \hat{W}_1 \neq w |W_1 = w)$ denote the decoding error probability when  message $w,w=0, \ldots, M_n$ was transmitted.
Then, $P_1$ is given by
\begin{align}
P_1 & = (1-\alpha_n)\Pr( \hat{W}_1 \neq 0) + \frac{\alpha_n}{M_n} \sum_{w=1}^{M_n} \Pr( \hat{W}_1 \neq w |W_1 = w) \notag \\
& \leq \Pr( \hat{W}_1 \neq 0|W_1=0) + \frac{1}{M_n} \sum_{w=1}^{M_n} \Pr( \hat{W}_1 \neq w  | W_1 = w). \label{Eq_err_prob_uppr}
\end{align}
If $W_1=0$, then an error occurs  if $Y_{11} > \frac{\sqrt{tE_n}}{2}$. So, we have
\begin{align}
\Pr( \hat{W}_1 \neq 0|W_1=0) & = Q\left( \frac{\sqrt{tE_n}}{2} \right). \label{Eq_err_prob_uppr2}
\end{align}
Let $\cE_{11}$ denote the event $Y_{11} \leq \frac{\sqrt{tE_n}}{2}$ and $D_w$ denote the  error event in decoding message $w$ for the transmission scheme described in Section~\ref{Sec_proof_ortho_access} when the user is known to be active. Then, for every $w=1,\ldots,M_n$
\begin{align}
\Pr( \hat{W}_1 \neq w |W_1 = w) & =   \Pr(\cE_{11} \cup  \{ \cE_{11}^c \cap \hat{W}_1 \neq w \}| W_1 = w) \notag  \\
	& \leq  \Pr(\cE_{11}| W_1 = w)  + \Pr(  \cE_{11}^c | W_1 = w)  \Pr(  \hat{W}_1 \neq w | W_1 = w, \cE_{11}^c  ) \notag \\
	& \leq   \Pr(\cE_{11}| W_1 = w)  +  \Pr( D_w | W_1 = w) \notag
\end{align}
where the last step follows because $\Pr(\cE_{11}^c|W_1=w)\leq 1$ and by the definition of $D_w$. 

We  next define $\Pr(D) = \frac{1}{M_n} \sum_{w=1}^{M_n} \Pr(D_w)$. Since $P(\cE_{11} | W_1 =w)  = Q\left( \frac{\sqrt{tE_n}}{2} \right)$,
it  follows from~\eqref{Eq_err_prob_uppr} that
\begin{align}
P_1 & \leq 2 Q\left( \frac{\sqrt{tE_n}}{2} \right)+ P(D). \label{Eq_singl_usr_uppr}
\end{align}
We next upper-bound $P(D)$. To this end, we use the following upper bound on the average probability of error $P(\mathcal{E})$ of the Gaussian point-to-point channel for a code of blocklength $n$ with power $P$~\cite[Section~7.4]{Gallager68}
\begin{align}
P(\cE) & \leq  M_n^{ \rho}  \exp[-nE_0(\rho, P)],  \; \mbox{ for every } 0< \rho \leq 1 \label{Eq_upp_dec_AWGN} 
\end{align}
where 
\begin{align}
E_0(\rho, P) & \triangleq \frac{\rho}{2} \ln \left(1+\frac{2P}{(1+\rho)N_0}\right). \notag
\end{align} 
 By substituting in~\eqref{Eq_upp_dec_AWGN} $n$ by $\frac{n}{\el_n} - 1$ and  $P$ by $P_n = \frac{(1-t)E_n}{\frac{n}{\el_n} -1}$, we obtain that $P(D)$ can be upper-bounded in terms of the rate per unit-energy $\CR=\frac{\log M_n}{E_n}$ as follows:
\begin{align}
P(D) & \leq  M_n^{ \rho}  \exp\left[-\left(\frac{ n}{\el_n}-1\right)E_0(\rho, P_n)\right] \nonumber \\
& = \exp\left[   \rho  \ln M_n -  \left(\frac{ n}{\el_n}-1\right) \frac{\rho}{2} \ln \left(1+\frac{ 2E_n(1-t)}{ \left(\frac{ n}{\el_n}-1\right)(1+\rho)N_0}\right) \right] \nonumber \\
& = \exp\left[ -E_n(1-t) \rho   \left(  \frac{\ln \left(1+\frac{ 2E_n(1-t)}{ \left(\frac{ n}{\el_n}-1\right)(1+\rho)N_0}\right)}{ \frac{2E_n(1-t)}{ \left(\frac{ n}{\el_n}-1\right)} } -\frac{\CR}{(1-t) \log e} \right)\right]. \label{Eq_err_uppr}
\end{align}

We next choose $E_n = c_n \ln n$ with $c_n \triangleq \ln\bigl(\frac{n}{\el_n\ln n}\bigr)$. Since, by assumption, $\el_n = o(n / \log n)$, this implies that $\frac{\el_nE_n}{n} \to 0$ as $n \to \infty$, hence 
$\frac{E_n}{n/\el_n -1} \to 0$. Thus,
the first term in the inner most bracket in \eqref{Eq_err_uppr} tends to $1/((1+\rho)N_0)$ as $n \to \infty$. It follows that for $\CR < \frac{\log e}{N_0}$, there exists a sufficiently large $n'_0$, a $ t > 0$, a $\rho > 0$, and a $\delta>0$ such that, for $n\geq n'_0$, the RHS of \eqref{Eq_err_uppr} is upper-bounded by $\exp[-E_n(1-t) \rho \delta]$. It follows that, for our choice $E_n=c_n\ln n$, we have for $n\geq n_0'$
\begin{align}
P(D) & \leq \exp \left[ \ln \left(\frac{1}{n}\right)^{c_n\delta \rho(1-t)} \right]. \notag
\end{align}
 Since $c_n \to \infty $ as $n\to\infty$, and hence also
$c_n\delta \rho(1-t) \to \infty$, this yields
\begin{align}
P(D) & \leq \frac{1}{n^2} \label{Eq_act_dec_uppr}
\end{align}
for sufficiently large $n \geq n_0'$.

Similary, for $n\geq \tilde{n}_0$ and sufficiently large $\tilde{n}_0$, we can upper-bound
\begin{equation}
2 Q\left(\frac{\sqrt{t E_n}}{2}\right) \leq \frac{1}{n^2} \label{Eq_usr_det_uppr}
\end{equation}
by upper-bounding the $Q$-function as $Q(\beta)\leq \frac{e^{-\beta^2/2}}{\sqrt{2\pi}\beta}$ and evaluating the resulting bound for $E_n=c_n\ln n$.
Using~\eqref{Eq_act_dec_uppr} and~\eqref{Eq_usr_det_uppr} in~\eqref{Eq_singl_usr_uppr}, we obtain for $n \geq \max(\tilde{n}_0,n_0')$ that
\begin{align}
P_1 \leq \frac{2}{n^2}. \notag
\end{align}
This proves Lemma~\ref{Lem_detect_uppr}.